\documentclass[11pt]{article}
\usepackage{amsmath,amssymb}

\usepackage[margin = 1.0 in]{geometry}

\usepackage{wrapfig,booktabs}

\usepackage[utf8]{inputenc}
\usepackage[english]{babel}
\usepackage{comment}
\usepackage{amsthm}
\usepackage{hyperref}
\usepackage[capitalise]{cleveref}
\usepackage{lipsum}

\usepackage{algorithm}
\usepackage[noend]{algpseudocode}

\usepackage{graphicx}
\graphicspath{ {images/} }
\usepackage{amscd}
\usepackage[all]{xy}

\numberwithin{equation}{section}

\newtheorem{fact}[equation]{Fact}
\newtheorem{lemma}[equation]{Lemma}
\newtheorem{theorem}[equation]{Theorem}
\newtheorem{definition}[equation]{Definition}

\newtheorem{corollary}[equation]{Corollary}

\newtheorem{claim}[equation]{Claim}
\newtheorem{remark}[equation]{Remark}
\newtheorem*{remark*}{Remark}

\newtheorem{observation}[equation]{Observation}
\newtheorem{question}{Question}

\usepackage{xcolor}
\hypersetup{
    colorlinks,
    linkcolor={blue!80!black},
    citecolor={green!50!black},
    urlcolor={red!40!black}
}

        {\medskip}

        {\hspace*{\fill}$\Box$\par\vspace{4mm}}
        {\hspace*{\fill}$\Box$\par}

\newcommand{\Z}{\mathbb{Z}}
\newcommand{\N}{\mathbb{N}}
\newcommand{\Q}{\mathbb{Q}}

\newcommand{\C}{\mathbb{C}}

\newcommand{\omegap}{\omega_p}

\newcommand{\omegan}{\omega_n}

\newcommand{\Gal}{\mathsf{Gal}}
\newcommand{\E}{\mathbb{E}}
\newcommand{\cov}{\mathsf{cov}}
\newcommand{\calC}{\mathcal{C}}
\newcommand{\calA}{\mathcal{A}}
\newcommand{\poly}{\mathsf{poly}}
\newcommand{\Aut}{\mathsf{Aut}}
\newcommand{\wt}{\mathsf{wt}}
\newcommand{\lcmG}{\Delta}

\newcommand{\G}{\mathcal{G}}
\newcommand{\pmo}{\{-1,+1\}}

\numberwithin{equation}{section}

\DeclareMathOperator{\supp}{supp}

\DeclareMathOperator{\Span}{Span}

\DeclareMathOperator{\Codim}{Codim}

\DeclareMathOperator{\AND}{AND}
\DeclareMathOperator{\Var}{Var}
\DeclareMathOperator{\Cov}{Cov}
\DeclareMathOperator{\AT}{\mathsf{AT}}
\DeclareMathOperator{\LCM}{LCM}
\DeclareMathOperator{\Ann}{Ann}

\newcommand{\pdc}[1]{\textcolor{red}{Pranjal:#1}}

\begin{document}
\title{
Structure of sparse Boolean functions over Abelian groups, and its application to testing\footnote{This work generalizes our MFCS 2024 paper, \emph{``On Fourier Analysis of Sparse Boolean Functions over Certain Abelian Groups''}, in which we studied sparse Boolean functions on groups of the form $\Z_{p_1}^{n_1} \times \cdots \times \Z_{p_t}^{n_t}$ with distinct primes $p_i$.
}
}

\author{
Sourav Chakraborty\footnote{Indian Statistical Institute, Kolkata, India}
\and
Swarnalipa Datta\footnotemark[1]
\and
Pranjal Dutta\footnote{Nanyang Technological University, Singapore. Pranjal Dutta is supported by the SUG Grant (\#025774-00001) titled~``Deordering and Derandomization in Algebraic Complexity'', funded by Nanyang Technological University.}
\and
Arijit Ghosh\footnotemark[1]
\and
Swagato Sanyal\footnote{University of Sheffield, Sheffield, UK}
}



\date{}

\maketitle

\begin{abstract}
We study Fourier-sparse Boolean functions over general finite Abelian groups. 
A Boolean function $f : \G \to \{-1,+1\}$ is $s$-sparse if it has at most $s$ non-zero Fourier coefficients. 
We introduce a general notion of \emph{granularity} of Fourier coefficients and prove that every non-zero coefficient of an $s$-sparse Boolean function has magnitude at least 
\[
\frac{1}{2^{\varphi(\lcmG)/2} \, s^{\varphi(\lcmG)/2}},
\]  
where $\Delta$ denotes the exponent of the group $\G$ (that is, the maximum order of an element in $\G$) and $\varphi$ is the Euler's totient function.
This generalizes the celebrated result of Gopalan et al. (SICOMP 2011) for $\mathbb{Z}_2^n$, extending it to all finite Abelian groups via new techniques from group theory and algebraic number theory.

Using our new structural results on the Fourier coefficients of sparse functions, we design an efficient sparsity testing algorithm for Boolean functions. The tester distinguishes whether a given function is $s$-sparse or $\epsilon$-far from every $s$-sparse Boolean function, with query complexity $\poly\left((2s)^{\varphi(\lcmG)},1/\epsilon \right)$.
In addition, we generalize the classical notion of Boolean degree to arbitrary Abelian groups and establish an $\Omega(\sqrt{s})$ lower bound for adaptive sparsity testing.
\end{abstract}

\section{Introduction}
Boolean functions are fundamental objects of study in computer science. For a discrete domain $\mathcal{D}$, a Boolean function $f:\mathcal{D}\to\{+1, -1\}$ models a decision task where each member of $\mathcal{D}$ is classified into one of two classes. Boolean functions play a vital role in the study of digital circuits and computer hardware. They are also significant in the study of algorithms and complexity, particularly in problems where the set $\mathcal{D}$ of instances is endowed with an algebraic structure. Examples of such problems include matrix multiplication and polynomial evaluation.


The case of Boolean function complexity with $\mathcal{D}=\Z_2^n$ has been widely studied. These functions are often analyzed in connection with their Fourier transform (see Section~\ref{sec:fourier_basics}) and a significant amount of research has focused on the structural properties of Fourier spectra of important classes of these functions. One such class that this work focuses on is that of Fourier-sparse functions. These are functions with only a few non-zero Fourier coefficients, formally defined in~\cref{def:fourier-sparsity}. We will denote by $s_f$ the Fourier sparsity of a Boolean function $f$. Fourier sparsity and Fourier-sparse functions  have known connections with a variety of areas of Boolean function analysis and computational complexity like property testing \cite{gopalan2011testing}, learning theory \cite{haviv2016list, arunachalam2021two}, distance estimation~\cite{yaroslavtsev2020fast} and communication complexity \cite{MO09, TsangWXZ13, HHL18, Sanyal19, eccc/MandeS20}. These connections provide enough motivation to comprehend the structure of Fourier coefficients for Fourier-sparse Boolean functions.

In this work, we extend the study of Fourier-sparse Boolean functions to the domains $\mathcal{D}$, which are finite Abelian groups.
Boolean functions over general Abelian groups have been studied in both mathematics and computer science. A celebrated result regarding such functions is Chang's Lemma \cite{Chang02}. Chang's lemma over~$\Z_2^n$ has found numerous applications in complexity theory and algorithms \cite{Ben-SassonRTW14, ChanLRS16}, analysis of Boolean functions \cite{green2008boolean, TsangWXZ13}, communication complexity \cite{TsangWXZ13, DBLP:conf/coco/HosseiniLY19}, extremal combinatorics \cite{FriedgutKKK18}, and many more. Recently, \cite{fsttcs/0001MMMPS21} improved Chang's lemma over $\Z_2^n$ for some special settings of parameters, where Fourier sparsity played a crucial role. One motivation to study Fourier sparsity over a broader class of Abelian groups is to investigate possible generalizations of their bounds to those groups.

\medskip
{\bf \noindent Fourier analysis over finite Abelian groups in cryptography.}~For the past three decades, the field of cryptography has been utilizing concepts derived from Fourier analysis, specifically over finite Abelian groups.  Akavia, Goldwasser, and Safra~\cite{akavia2003proving} have combined some of these concepts to develop a comprehensive algorithm that can detect ``large'' Fourier coefficients of any concentrated function on finite Abelian groups, and compute a {\em sparse approximation} for the same. This algorithm has gained significant attention within the cryptography community, especially regarding the notion of ``bit security'' of the {\em discrete logarithm problem} (DLP), RSA, and learning with errors (LWE) problems; see~\cite{regev2009lattices,DBLP:journals/cjtcs/GalbraithLS18,akavia2008learning}. In particular, the ``nice'' structural results on the Fourier coefficients of a Boolean-valued function over the general Abelian group are of utmost importance and interest from a crypto-theoretic point of view. 

Interestingly, there are strong relationships between learning, sparsity, and sampling, in the context of Fourier-sparse Boolean functions. They have been rigorously studied in~\cite{sitharam1994pseudorandom,sitharam2000sampling,sitharam2001derandomized}. In~\cite{sitharam1994pseudorandom}, the authors  asked the following: 

\begin{question}\label{qn-1}
  What can be said about the structure of the Fourier coefficients of a Boolean function $f$ over a finite Abelian group $\G$, 
  where the support is significantly smaller compared to $\G$?  
\end{question}


Gopalan et al.~\cite{gopalan2011testing} proved that any non-zero Fourier coefficient of a Boolean function over $\Z_2^n$ with Fourier sparsity at most $s_f$, is {\em at least} $\frac{1}{s_f}$ in its absolute value. This gave a satisfactory answer of~\cref{qn-1} over $\Z_2^n$. Furthermore, they proved {\em robust versions} of their result for functions which are {\em approximately Fourier-sparse}. Finally, their structural results were used to design a sample-efficient algorithm to test whether a function is Fourier-sparse. 


In our work, we undertake the same task for Boolean functions over finite Abelian groups $\G$. Any finite Abelian group $\G$ can be written in the form $\mathbb{Z}_{{p_1}^{m_1}} \times \dots \times  \mathbb{Z}_{{p_T}^{m_T}}$, ($p_i, \ i \in [T]$ are primes and not necessarily distinct). We prove lower bounds on the absolute value of any non-zero coefficient in terms of $s_f$ and $\lcmG$, where $\lcmG = \LCM \{p_1^{m_1}, \ldots, p_T^{m_T}\}$. Chakraborty et al.~\cite{DBLP:conf/mfcs/0001DDGS24} have already proven these results for Abelian groups of the form $\mathbb{Z}_{p_1}^{n_1} \times \cdots \times \mathbb{Z}_{p_T}^{n_T}$, where $p_i, \ i \in [T]$ are distinct primes, which is a generalization of Gopalan et al.~\cite{gopalan2011testing}. Observe that $\mathbb{Z}_p^n$ for $p$ prime is a vector space. An important thing to note here is that a finite Abelian group is not a vector space in general. Since most of the properties of vector space are lost in case of $\G$, we use different techniques to prove our results. They also showed a tightness result that complements our lower bounds. In particular, our bound implies that a lower bound of $\frac{1}{\Theta(s_f)}$ that~\cite{gopalan2011testing} showed {\em does not} hold anymore for $\G$. Finally, we use our bounds to design a testing algorithm for Fourier-sparse Boolean functions over $\G$.

\paragraph*{Why care about Fourier-sparse Boolean functions over Abelian groups?} There has been a considerable amount of interest in studying the complexity of reconstructing or learning functions of the form $f : \mathcal{D} \to \C$, where $\mathcal{D}$ is a known domain (more general than a hypercube, such as a general finite Abelian group), and $f$ is Fourier-sparse; see~\cite{sitharam2001derandomized,rudelson2008sparse,morotti2013reconstruction,chen2020reconstruction,yang2023lower}. Fourier-sparse functions over various finite Abelian groups have gained much interest with the advancement in sparse Fourier transform algorithms \cite{hassanieh2012nearly,hassanieh2012simple,bafna2018thwarting}. These algorithms improve the efficiency of the standard Fast Fourier Transform algorithms by taking advantage of the sparsity itself.  To reliably use sparse Fourier transform algorithms, it is beneficial to have a way to {\em test} if a function is $s$-sparse or, more generally, to {\em estimate} the distance of a function to the closest $s$-sparse function.  In this work, we consider the problem of non-tolerant sparsity-testing of Boolean functions over finite Abelian groups.

Finally, apart from mathematical curiosity and potential cryptographic applications (as mentioned previously), structural results on Fourier-sparse functions $f: \Z_N \to \C$, for some $N \in \N$, have also found algorithmic applications in SOS-optimization and control theory. These applications have further implications in certifying maximum satisfiability (MAX-SAT) and maximum k-colorable subgraph (MkCS) problems; see~\cite{fawzi2015sparse,yang2022short,yang2022computing}.

\subsection{Our results}

Throughout the article, we will be working with finite Abelian groups. Any finite Abelian group $\G$ can be written in the form $\G := \mathbb{Z}_{{p_1}^{m_1}} \times \dots \times \mathbb{Z}_{{p_T}^{m_T}}$, where $p_i$ are primes and not necessarily distinct. Let $f: \G \to \{-1, +1\}$, and $f(x) = \sum_{\chi\in \hat{\G}} \widehat{f}(\chi)\chi(x)$ be the Fourier transform of $f$, where $\hat{\G}$ is the set of characters of the Abelian group $\G$. Also, let $\lcmG = \LCM\{p_1^{m_1}, \ldots, p_T^{m_T}\}$.

We say a Boolean function is $s$-{\em sparse}, if it has at most $s$ non-zero Fourier coefficients. \cite{gopalan2011testing} proved that for any $s$-sparse Boolean functions over $\Z_2^n$, the magnitude of the Fourier coefficients are $k$-granular where~$k=\lceil\log_2 s\rceil +1$. A real number is $k$-{\em granular} if it is an integer multiple of $1/2^k$. One wonders whether such a phenomenon still holds over a more general group $\G$.


This notion of granularity made sense over $\Z_2^n$, since in this case, all the Fourier coefficients are {\em rational} (and hence real) numbers. But when the domain of the function is a general group~$\G$, the Fourier coefficients are necessarily complex numbers. So, we would like 
to suitably define granularity, and show that such a property still holds for $s$-sparse Boolean-valued functions over $\G$. Our first conceptual contribution in this paper is to generalize the notion of granularity appropriately.



\begin{definition}\label{defn_granularity}
    \textbf{(Granularity)} A complex number is said to be $k$-granular or has granularity $k$ with respect to $\G$ if it is of the form $\frac{g(\omega_{\lcmG})}{\lcmG^{k}}$, where $g(X)\in \mathbb{Z}[X]$ 
    and $\omega_{\lcmG}$ is a primitive $(\lcmG)^{th}$ root of unity. 

\end{definition}

Note that, this goes well with the definition of granularity in \cite{gopalan2011testing} for the case of $\mathbb{Z}_2$ as $\omega_2$ is either $+1$ or $-1$ and hence $g(\omega_2)$ is an integer for any $g(X) \in \mathbb{Z}[X]$.
    
    




We will also need a {\em robust} version of the definition of granularity of a complex number. 

\begin{definition}\label{defn_mu_close_to_granular}
    \textbf{($\mu$-close to $k$-granular)} A complex number $v$ is said to be $\mu$-close to $k$-granular with respect to $\G$ if $|v-\frac{g(\omega_{\lcmG})}{\lcmG^{k}}| \leq \mu$, for some non-zero polynomial $g(X)\in \mathbb{Z}[X]$.

\end{definition}

Now, we are ready to formally state our two main structural results. All our results hold for Boolean-valued functions over the more general Abelian group $\G$. 


Our first theorem says that for any Boolean-valued function over $\G$ that is {\em close} to being sparse, all its large Fourier coefficients are close to being granular. This is a generalization of the structural theorem of~\cite[Theorem 3.3]{gopalan2011testing}, which was proved over~$\mathbb{Z}_2^n$.

\begin{theorem}[Structure theorem 1]\label{lem:structure}
Let $f:\G\to \{-1,+1\}$ be a Boolean-valued function and let $B$ be the set of characters corresponding to the set of $s$-largest Fourier coefficients of $f$ (in terms of magnitude). If $\sum_{\chi \in B}|\widehat{f}(\chi)|^2 \geq (1-\mu)$ then for all $\chi \in B$, $\widehat{f}(\chi)$ is $\frac{\mu}{\sqrt{s}}$-close to $k$-granular, where $k= \lceil \log_{\lcmG} 2s \rceil$.

\end{theorem}



For a function $f:\G\to \{-1,+1\}$, $\sum_{\chi \in B}|\widehat{f}(\chi)|^2 \geq (1-\mu)$ (where $B$ is the set of characters corresponding to the set of $s$ largest coefficients of $f$) can also be stated as ``there is an $s$-sparse function $g:\G\to \mathbb{C}$ with the $\ell_2$-distance between $f$ and $g$ is at most $\sqrt{\mu}$''.  But note that this {\em does not} guarantee that there is an~$s$-sparse Boolean-valued function $g:\G\to \{-1,+1\}$ with $\ell_2$-distance between $f$ and $g$ being at most $\sqrt{\mu}$. However, our second theorem proves that one can indeed find an~$s$-sparse Boolean-valued function in a close enough vicinity, thus generalizing~\cite[Theorem 3.4]{gopalan2011testing}. 

\begin{theorem}[Structure theorem 2]\label{structure_theorem_2}
    Let $f:\G \to \{-1,+1\}$ be a Boolean-valued function and let $B$ be the set of characters corresponding to the set of $s$-largest Fourier coefficients of $f$ (in terms of magnitude). If $\sum_{\chi \in B}|\widehat{f}(\chi)|^2 \geq (1-\mu)$, with $\mu \leq \frac{1}{8 \times 2^{\varphi(\lcmG)} s^{\varphi(\lcmG)}}$, then there exists an $s$-sparse Boolean-valued function $F:\G \to \{-1,+1\}$ with $\ell_2$ distance between $f$ and $F$ is at most $\sqrt{2}\mu$. 
\end{theorem}

One important corollary of \cref{lem:structure} is that for any $s$-sparse Boolean function $f:\mathbb{Z}_2^n\to \{-1,+1\}$ the non-zero Fourier coefficients has magnitude at least $1/2^k$, where $k = \lceil \log_2 s\rceil+1$.  Unfortunately, such a simple corollary cannot be claimed for $s$-sparse functions $f:\mathbb{Z}_{p^m}^n\to \{-1,+1\}$ if $p^m \neq 2$. The main reason is that the definition of granularity is for complex numbers, rather than real numbers and hence such a lower bound cannot be directly deduced. However, borrowing results from algebraic number theory, we can obtain a lower bound on the Fourier coefficients of Boolean-valued functions from $\G$ to $\{-1,+1\}$, for any finite Abelian group $\G$.

\begin{theorem}[Fourier coefficient lower bound]\label{granularity_of_Z_p}
Let $f:\G \to \{-1,+1\}$, with Fourier sparsity $s_f$. Then, for any $\chi \in \supp (f)$, we have~$|\hat{f}(\chi)| \geq \frac{1}{2^{\varphi(\lcmG)/2} s^{\varphi(\lcmG)/2}}$.
\end{theorem}

\begin{remark*}
One can also prove a weaker lower bound of the form $\frac{1}{2^{s_f/2} s_f^{s_f/2}}$, which is $\lcmG$-independent; for details, see Theorem~\ref{thm:p-independent-lb}.
\end{remark*}

Observe that the lower bound in~\cref{granularity_of_Z_p} is much lower than~$1/s_f$. One may wonder how tight our result is. It is known that $1/s$ is tight for the case when $p=2$. For example, consider the function~$AND:\mathbb{Z}_2^n \to \{-1,+1\}$. Its non-empty Fourier coefficients are either $\frac{1}{2^{n-1}}$ or $-\frac{1}{2^{n-1}}$, while the empty (constant) coefficient being $1-\frac{1}{2^{n-1}}$. 

To our pleasant surprise, we construct $s$-sparse Boolean-valued functions over $\Z_p^n$, for $p \ge 5$, such that they have nonzero Fourier coefficients with absolute value being~$o(1/s)$.

\begin{theorem}[Small Fourier coefficients] \label{thm:small-fourier}
      For every prime $p\geq 5$, and large enough $n$, there exist a positive constant $\alpha_{p}$ that depends only on $p$ and 
      a function $f:\mathbb{Z}_p^n\to \{-1,+1\}$ with Fourier sparsity $s_f$ 
      satisfying the following property: 
      $$
        \min_{\chi \in \supp(f)} \left| \widehat{f} (\chi)\right| \leq 1/s_f^{1+\alpha_{p}}.
      $$
\end{theorem}

We prove a generalized version of the lower bound result (\cref{granularity_of_Z_p}) for Boolean-valued functions  over $\G$.

Finally, we design efficient algorithms for testing whether a function $f:\G \to \{-1, +1\}$ is $s$-sparse or ``far'' from $s$-sparse-Boolean.

To state our results we need to define what we mean by a function $f:\G\to \{-1, +1\}$ is $\epsilon$-far from $s$-sparse. 

\begin{definition}
A function $f:\G\to \{-1, +1\}$ is $\epsilon$-far from $s$-sparse-Boolean if for every $s$-sparse function $g:\G\to \{-1, +1\}$ the $\ell_2$-distance of $f$ and $g$ is at least $\sqrt{\epsilon}$.\footnote{In property testing usually the distance measure used is Hamming distance between two Boolean functions. But since we are using $\ell_2$ distance in our other theorem, so for ease of presentation, we have defined the farness in terms of $\ell_2$ instead of Hamming distance. Also, note that for a pair of Boolean-valued functions the square of the $\ell_2$ distance and Hamming distance are the same up to a multiplicative factor of $4$, see~\cref{lemma_Boolean_distance_equivalence}.} 
\end{definition}





We say that an algorithm (property tester) $\calA$~~$\epsilon$-tests $\calC$, for a class of functions $f:\G \to \{-1,+1\}$, if given access to the truth table of a function $f$,  whether $f \in \calC$, or $f$, is ``$\epsilon$-far from $\calC$'' can be tested using $\calA$ with success probability (called the {\em confidence}) $\ge 2/3$. The number of queries to the truth-table of $f$ made by $\calA$ is called the query complexity of $\calA$. 

Using the structure theorems (\cref{lem:structure} and \cref{structure_theorem_2}), we prove the following theorem which tests sparsity of a Boolean-valued function $f:\G \to \{-1,+1\}$.

\begin{theorem}[Testing $s$-sparsity]\label{thm:algo-sparse-check}
There is a non-adaptive $\poly(s, 1/\epsilon)$ query algorithm with confidence $2/3$, which tests whether a given function $f: \G \rightarrow \{-1,+1\}$, is $s$-sparse or $\epsilon$-far from $s$-sparse-Boolean. 
\end{theorem}


We complement our result by showing a query-complexity lower bound for sparsity-testing algorithms. Gopalan et al.~\cite{gopalan2011testing} gave a $\Omega(\sqrt{s})$ lower bound for $s$-sparsity testing algorithms over $\Z_2^n$. An important component of their proof was to cleverly use an alternative notion of {\em degree} (borrowed from~\cite{BernasconiC99}) of a Boolean function over $\Z_2$. We also give a similar lower bound over $\G$, by appropriately generalizing the useful notion of degree. For details on the definition of degree, see~proof idea of \cref{thm_testing_lower_bound} in \cref{sec:pf-idea} and in \cref{section_deg_p}, \cref{section_algo_lower_bound}.

\begin{theorem}\label{thm_testing_lower_bound}
    For Boolean valued functions on $\G$, to adaptively test $s$-sparsity, the query lower bound of any algorithm is $\Omega(\sqrt{s})$.
\end{theorem}
\cref{thm_testing_lower_bound} can be generalized for Boolean valued functions on $\G$, which will give us the same lower bound.

\subsection{Proof ideas} \label{sec:pf-idea}

In this section, we briefly outline the proof ideas of our main theorems. Note that $\mathbb{Z}_p^n$ is a vector space for prime $p$.
For general Abelian groups it can be shown that vector space structure does not exists.  To deal with this issue we define an operation $*:\G\times \G \to \mathbb{Z}_{\Delta}$, which is a generalization of the inner product when $\G = \mathbb{Z}_p^n$. Note, that the operation, `*', cannot be an inner-product in general as $\G$ does not have a vector space structure. We call it the ``pseudo inner-product'' operation (see Definition~\ref{defi:pseudo-inner product}), and using this pseudo inner-product operation we define certain normal subgroups that acts like subspaces of $\G$, and their cosets are used to partition $\G$. Details are given in Section~\ref{section_Preliminaries}.
Let us first sketch the proof of \cref{lem:structure}.

\medskip
{\bf \noindent Proof idea of Theorem~\ref{lem:structure}.}~Our goal is to show that if $f$ is $\mu$-close to some $s$-sparse complex-valued function in $\ell_2$, then there exists a non-zero polynomial $g(X) \in \mathbb{Z}[X]$ such that the following properties hold. 
\begin{enumerate}
    \item The sum of the absolute values of its coefficients is {\em at most} $\lcmG^k$, where $k :=\lceil \log_\lcmG 2s \rceil$.
    \item The {\em distance} between the absolute value of each non-zero Fourier coefficient of $f$ and $|g(\omega_\lcmG)|/\lcmG^k$ is at most $\mu/\sqrt{s}$.
\end{enumerate}
To show the above, we first utilize a probabilistic method 
to prove that for each character $\chi_i \in B$ in the Fourier support of $f$, there exists a matrix $A \in \mathbb{Z}_\lcmG^{k \times T}$, and a column vector $b \in \mathbb{Z}_\lcmG^{T \times 1}$, such that -- (1) $\chi_i$ is a solution of the system of equations $A\chi = b$, and (2) {\em no other} character in the Fourier support of $f$ is a solution of $A\chi = b$, where $B$ be the set of $s$-largest Fourier coefficients of $f$. Using our pseudo inner product structure, we establish the existence of such $A$ and $b$, we consider the Fourier transform of the projection (Definition~\ref{def:proj-op}) for the solution space of $A\chi=b$ (see Definition~\ref{def:proj-op} and Lemma~\ref{lemma_granularity_Z_p}). The projection operator, as the name suggests, is an operator that projects $\G$ onto a coset which yields a partition of the Fourier spectrum of $f$. Then we show that the $\ell_2$ Fourier weight of $S\cap H$, i.e.,~$\sum_{\chi\in S\cap H} |\widehat{f}(\chi)|^2$ is upper bounded by $\mu/\sqrt{s}$, where $S=\overline{B}$, and $H$ is a coset of $A^\perp$ that are solutions to the system of linear equations $A\chi = b$.  For details, see \cref{sec:lb-z_p}.

\medskip
{\bf \noindent Proof idea of Theorem~\ref{granularity_of_Z_p}.}~If we put $\mu=0$ in \cref{lem:structure}, we get that there exists a $g \in \Z[X]$, with the sum of the absolute values of its coefficients is {\em at most} $\lcmG^k$, such that $|\widehat{f}(\chi)| \ge |g(\omega_\lcmG)/\lcmG^k|$, where $k = \lceil \log_{\lcmG} 2s_f\rceil$, $s_f$ being the sparsity of the Boolean valued function $f$.  The remaining part of the proof is to show that for any polynomial $g$ with the aforementioned properties, $|g(\omega_\lcmG)|/\lcmG^k\,\geq\, 1/(2s_f)^{\lceil \varphi(\lcmG)/2 \rceil}$. 

As stated earlier, we use a non-trivial result from algebraic number theory (Theorem~\ref{theorem_roots_bound}), which states that if $f\in \mathbb{Z}[x]$, such that $f(\omega_n)\neq 0$, where $\omegan$ be a primitive root of unity, then, $|\prod_{i \in \Z_n^{\ast}} f(\omegan^i)| \ge 1$. We also use the fact that the sum of the absolute values of the coefficients of $g$ is at most $\lcmG^k$, to get an upper bound on the quantities $|g(\omega_\lcmG^i)|$, for any $i \in \{ a \in \mathbb{Z}_\lcmG : \gcd(i,a) =1 \}$. Combining these two facts, we obtain our lower bound; for details see~Section~\ref{sec:lb-z_p}.

\medskip
{\bf \noindent Proof idea of Theorem~\ref{structure_theorem_2}.}~We first show that the given function $f:G \to \{-1,+1\}$, that is $\mu$-close to some $s$-sparse complex-valued function in $\ell_2$, can be written as the sum of two functions $F$ and $G$, where the Fourier coefficients of $F$ are $\lceil \log_\lcmG 2s \rceil$-granular and the absolute value of the Fourier coefficients of $G$ are upper bounded by $\mu/\sqrt{s}$. This follows from \cref{lem:structure}. Then we show that the range of $F$ is $\{-1,+1\}$, which uses the following facts:
\begin{enumerate}
    \item $(F+G)^2=f^2=1$ and
    \item $F^2$ is $2\lceil \log_\lcmG 2s \rceil$-granular.
\end{enumerate}
We compute $\mathbb{E}[G(x)^2]$ in order to find an upper bound on the Fourier coefficients of $H:=G(2f-G)$, which helps us to conclude that $\widehat{F^2}(\chi)=0$ for all $\chi \neq \chi_0$, and $\widehat{F^2}(\chi_0)=1$, where $\chi_0$ is the character which takes the value $1$ at all points in $\G$. Then we complete the proof by showing that the $\Pr_x[x\in \G | f(x) \neq g(x)]$ is $\leq \mu^2/2$, which implies that $F$ is $\sqrt{2}\mu$ close to $f$ in $\ell_2$ by \cref{lemma_Boolean_distance_equivalence}. This idea has also been employed in ~\cite{gopalan2011testing}.



\medskip

{\bf\noindent Proof idea of Theorem~\ref{thm:algo-sparse-check}.}~The main idea of the algorithm is to partition the set of characters into buckets (or cosets) using our pseudo inner product structure, and estimate the {\em weights} of the individual buckets (that is the sum of the squares of the absolute values of the Fourier coefficients corresponding to the characters in the buckets).  We know from \cref{granularity_of_Z_p} that all the coefficients of an~$s$-sparse function are at least as large as~$1/(2s)^{\lceil \varphi(\lcmG)/2 \rceil}$. So, we are certain that if the weights of the buckets can be approximated within an additive error or $\tau/3$, where $\tau \geq 1/(2s)^{\varphi(\lcmG)}$, then in the case the function is $s$-sparse, not more than $s$ of the buckets can have weight more than $\tau/2$. On the other hand, we will show that if the function $f$ is $\epsilon$-far from any $s$-sparse Boolean function then with a high probability at least $(s+1)$ buckets will have weight more than $\tau$, making the estimated weight at least $2\tau/3$; see~Lemma~\ref{lem:bucket}.

The challenge is that estimating the weights of the buckets is not easy if the characters the randomly partitioned into buckets. Here we use the ideas from Gopalan et al~\cite{gopalan2011testing} {and appropriately modify them to handle the technicalities} of working with $\G$. We choose the buckets carefully. The buckets corresponds to the cosets of $H^{\perp}$ in $\G$, where $H$ is a random subspace of $\G$ of dimension, $t = \Theta(s^2)$.
For such kinds of buckets, we show that estimation of the weight can be done using a small number of samples. We also need to use the concept of `random shift', see~\cref{def:random-shift}, to avoid the corner case of characters being put into the trivial bucket.
 
Unlike~\cite{gopalan2011testing}, it becomes a bit more challenging to prove that if $f$ is $\epsilon$-far from any $s$-sparse Boolean function over $\G$, then with high probability at least $(s+1)$ buckets will have weight more than $\tau$. Since we partition the set of characters by cosets, the events whether two characters land in the same bucket (that is same coset)  are {\em not independent} -- since two characters (in the case $\G \neq \mathbb{Z}_2^n$) can be scalar multiple of each other; this is where some additional care is required (which was not the case in~\cite{gopalan2011testing}). Under random shifts and case-by-case analysis, we can show that the two events are {\em not correlated}, i.e., the covariance of the corresponding indicator variables of the two events is $0$; see Lemma~\ref{thm:random-coset}. Thus, we can use Chebyshev's inequality and then Markov's inequality to bound the number of buckets that can be of low weight, or in other words prove that the number of ``heavy'' weight buckets is more than $s+1$. 

\medskip

{\bf \noindent Proof idea of Theorem~\ref{thm_testing_lower_bound}.}  In~\cite{gopalan2011testing}, Gopalan et al.~proved a query lower bound over $\Z_2^n$, by using a natural notion of {\em degree} of a Boolean function, denoted~$\deg_2$. They crucially used the fact that for a Boolean function $f$ over $\Z_2$, $2^{\dim(f)} \geq s_f \geq 2^{\deg_2(f)}$; this was originally proved in~\cite{BernasconiC99}. 
To define the degree, let us consider all possible restrictions $f|_{V_{b,r_1, \ldots,r_t}}$ of $f$, where~$V_{b,r_1, \ldots,r_t}$ is a coset of $V_{0,r_1, \ldots,r_t}$ in $\G$ as defined as
\[
 V_{b,r_1,\ldots,r_k} \;:=\; \{ x\in \G : r_j \cdot x \,=\, b_j~\pmod{\lcmG}\ \forall j \in [t]\}\;.
\]
Then the degree over $\G$ of $f$, denoted by $\deg_{\G}$, is defined in the following way.
    \begin{align*}
        \deg_{\G}(f) \;=\; \max_{\ell} \{\ell\;=\dim(V_{b,r_1, \ldots,r_t}) \;:\; s_{f|_{V_{b,r_1, \ldots,r_t}}} \;=\; \lcmG^{\dim(V_{b,r_1, \ldots,r_t})} \},
    \end{align*}
    where $s_{f|_{V_{b,r_1, \ldots,r_t}}}$ is the Fourier sparsity of the function $f|_{V_{b,r_1, \ldots,r_t}}$. \cite{gopalan2011testing} defined the degree over $\Z_2^n$ via similar restrictions. However, \cite{gopalan2011testing} argued that this is a natural definition of the degree over $\Z_2^n$. To argue, observe that~$\widehat{f} = \frac{1}{2^n} H_n f$, where $H_n$ is the $2^n \times 2^n$ Hadamard matrix, when $x_i \in \Z_2^n$ are seen in lexicographic order. Consider the restrictions $f|_{V_{b,r_1, \ldots,r_t}}$ that takes the value $1$ at those points such that each entry of $\widehat{f}|_{V_{b,r_1, \ldots,r_t}}$ is nonzero. Then $\deg_2$ can be defined as the dimension of $V_{b,r_1, \ldots,r_t}$ which is {\em largest} amongst them! In that case, all the Fourier coefficients of $f|_{V_{b,r_1, \ldots,r_t}}$ are nonzero.

Interestingly, we can also show that the above definition of $\deg_{\G}$ is natural, mainly because $\widehat{f} = \frac{1}{|\G|} V_\G f$, where $V_\ell$ is a $p_\ell \times p_\ell$ {\em Vandermonde} matrix $V_\ell$, whose $(i,j)$-th entry is $(V_\ell)_{i,j} := \omega_{p_\ell}^{(i-1)(j-1)}$, and $V_\G$ is defined by taking {\em Kronecker products} of $V_1, \ldots, V_T$, i.e.,~$V_\G := V_1 \otimes \cdots \otimes V_T$. Note that $H_T = V_T$, over $\Z_2^T$. Similar to~\cite{BernasconiC99}, one can also show that $\lcmG^{\dim(f)} \geq s_f \geq \lcmG^{\deg_\G (f)}$ (see \cref{lemma_dim_sparsity_deg}). This plays a crucial role in the proof. 



We first define two distributions $\mathcal{D}_{Yes}$ and $\mathcal{D}_{No}$ on the set of Boolean valued functions from $\G$ to $\{-1,+1\}$.
Let us choose a {\em random} $t$-dimensional coset $H$ of $\mathbb{Z}_\lcmG^{Ct}$, for some parameter $C$ (to be fixed later). Let~$\mathcal{C}$ be the set of all cosets of $H$. There are 2 main steps as follows.
\begin{enumerate}
    \item We construct random functions $f$ by making $f$ a constant on each coset of $\mathcal{C}$. The constant is chosen randomly from $\{-1,+1\}$. We call this probability distribution \textbf{$\mathcal{D}_{yes}$}.
    \item We choose a random function $f$ randomly from $\mathbb{Z}_\lcmG^{Ct}$, conditioned on the fact that $f$ is $2 - \tau$ far in $\ell_2$ from any function which has $\deg_\G = t$, where $\tau$ is as defined in \cref{thm:algo-sparse-check}. Let us call this distribution \textbf{$\mathcal{D}_{No}$}.
\end{enumerate}

We show that if an adaptive query algorithm makes less than $q< \Omega(\lcmG^{t/2})$ queries, then the total variation distance $||\mathcal{D}_{Yes} - \mathcal{D}_{No}||_{TV}$ between the two distributions $\mathcal{D}_{Yes}$ and $\mathcal{D}_{No}$ is $\leq \frac{1}{3}$. This proves that any adaptive query algorithm which distinguishes between $\mathcal{D}_{Yes}$ and $\mathcal{D}_{No}$, i.e.,~where $||\mathcal{D}_{Yes} - \mathcal{D}_{No}||_{TV} > \frac{1}{3}$, {\em must} make at least $\Omega(p^{t/2})$ queries.  This essentially proves \cref{thm_testing_lower_bound}. 

\section{Preliminaries} \label{section_Preliminaries}

\subsection{Fourier Analysis over 
\texorpdfstring{$\G$
}{}}
\label{sec:fourier_basics}

Any finite Abelian group $\G$ can be written as $\mathbb{Z}_{{p_1}^{m_1}} \times \dots \times  \mathbb{Z}_{{p_T}^{m_T}}$, where $p_i, \ i \in [T]$, are primes and not necessarily distinct. Throughout this paper, we will denote $\mathbb{Z}_{{p_1}^{m_1}} \times \dots \times  \mathbb{Z}_{{p_T}^{m_T}}$ by $\G$. So $\G$ is a finite Abelian group with $|\G|=p_{1}^{m_1} \cdots p_T^{m_T}$, where $|.|$ denotes the order of $\G$. We will denote this root of unity by $\omega_N$ a $N^{th}$ primitive root of unity, that is $e^{2\pi\iota/N}$. Let us begin by defining the characters of $\G$.

\begin{definition}\label{defn_character_1}
    \textbf{(Character)} For $\G = \mathbb{Z}_{p_1^{m_1}} \cdots \mathbb{Z}_{p_T^{m_T}}$ a character of the group $\G$ is a homomorphism $\chi: \G \to \mathbb{C}^\times$ of $\G$, that is, $\chi$ satisfies the following: for all $x,y \in \G$ we have $\chi(x+y)=\chi(x) \chi(y)$.

    \noindent Equivalently, a character of $\G$ is of the form $$\chi(x^{(1)},\dots, x^{(T)}) = \chi_{r^{(1)}}(x^{(1)}) \cdots \chi_{r^{(T)}}(x^{(T)}),$$ where $r^{(i)} \in \mathbb{Z}_{p_i^{m_i}}$ and $\chi_{r^{(i)}}$ is a character of $\mathbb{Z}_{p_i^{m_i}}$ and is defined by $\chi_{r^{(i)}}(x^{(i)}) = \omega_{p_i^{m_i}}^{r^{(i)} x^{(i)}}$.
    
    Thus fo r any $(r^{(1)}, \dots, r^{(T)})\in \G$ we can define a corresponding character of $\G$ as $\chi_{r^{(1)}, \dots, r^{(T)}}$ that on input  $x = (x^{(1)}, \dots, x^{(T)})$ is defined as
    \begin{align}\label{eq:character}
    \chi_{(r^{(1)}, \ldots, r^{(T)})}(x^{(1)}, \ldots, x^{(T)}) 
    &= \omega_{p_1^{m_1}}^{r^{(1)}\cdot x^{(1)}} \cdots \omega_{p_T^{m_T}}^{r^{(T)}\cdot x^{(T)}} \\
    &= \omega_N^{\sum_{i=1}^T r^{(i)}\cdot x^{(i)} \frac{\lcmG}{p_i^{m_i}} \pmod{\lcmG}},
\end{align}
where $\lcmG = \LCM \{ p_1^{m_1}, \ldots, p_T^{m_T} \}$.

\end{definition}

Now let us look at some properties of characters.

\begin{lemma}\label{lemma_character}
Let $\chi$ be a character of $\G$. Then,
\begin{enumerate}
    \item $\chi_0(x)=1$ for all $x\in \G$. 
    \item $\chi(-x)= \chi(x)^{-1}=\overline{\chi(x)}$ for all $x \in \G$. 
    \item For any character $\chi$ of $\G$, where $\chi\neq \chi_0$, $\sum_{x\in \G} \chi(x)=0$.
    \item $|\chi_0(x)|=1$ for all $x \in \G$.
\end{enumerate}
\end{lemma}

Now let us define the dual group of $\G$.

\begin{definition}
    \textbf{(Dual group)} The set of characters of $\G$ forms a group under the operation $(\chi\psi) (x) = \chi(x) \psi(x)$ and is denoted by $\widehat{\G}$, where $\chi$ and $\psi$ are characters of $\G$. $\widehat{\G}$ is called the dual group of $\G$.
\end{definition}

The following theorem states that $\G$ is isomorphic to its dual group.

\begin{theorem}
    $\widehat{\G} \cong \G$.
\end{theorem}

Let us look at the definition of Fourier transform for functions on $\G$.

\begin{definition} \label{def:foruier-transform}
    \textbf{(Fourier transform)} For any $\G = \mathbb{Z}_{p_1^{m_1}} \times \cdots \times \mathbb{Z}_{p_T^{m_T}}$ and any function $f: \G \to \mathbb{C}$, the Fourier transform $\widehat{f}: \widehat{\G}\to \mathbb{C}$ is 
    $$\widehat{f}(\chi_{r^{(1)},\ldots,r^{(t)}}) =\frac{1}{|\G|} \sum_{x\in \G} f(x) \omega_{p_1^{m_1}}^{-r^{(1)}\cdot x^{(1)}} \cdots \omega_{p_T^{m_T}}^{-r^{(T)}\cdot x^{(T)}},$$ where $x = (x^{(1)}, \dots, x^{(T)})$.

    The Fourier transform of a function $f:\G \to \mathbb{C}$ is defined as $\widehat{f}(\chi)= \frac{1}{|\G|} \sum_{x\in \G} f(x) \overline{\chi(x)}$, where $\overline{\chi(x)}$ is the conjugate of $\chi(x)$. 
\end{definition}


The following theorem states that any function from $\G$ to $\mathbb{C}$ can be written as a linear combination of characters of $\G$.

\begin{theorem}\label{Fourier_inversion}
\textbf{(Fourier inversion formula)} Any function $f: \G \to \mathbb{C}$ can be uniquely written as a linear combination of characters of $\G$, that is,
    \begin{equation} \label{eq:fourierex}
    f(x)= \sum_{\chi_{r^{(1)},\ldots,r^{(T)}} \in \widehat{\G}} \widehat{f}(\chi_{r^{(1)},\ldots,r^{(T)}}) \chi_{r^{(1)}, \dots, r^{(T)}}(x),\end{equation}

    where $x = (x^{(1)}, \dots, x^{(T)})$. 
\end{theorem}

\begin{theorem}\label{theorem_Parseval}
    \textbf{(Parseval)} For any two functions $f,g: \G \to \mathbb{C}$, $$\mathbb{E}_{x\in \G} [f(x)\overline{g(x)}] =\sum_{\chi\in \widehat{\G}} \widehat{f}(\chi) \overline{\widehat{g}(\chi)}.$$
    More specifically, if $f:\G \to \{-1,+1\}$ is a Boolean-valued function then
    $$
        \sum_{\chi \in \widehat{\G}} |\widehat{f}(\chi)|^2=1.
    $$
\end{theorem}


Now let us define the Fourier sparsity of a function $f$ on $\G$.

\begin{definition} \label{def:fourier-sparsity}
\textbf{(Sparsity and Fourier Support)} 
\begin{itemize}
    \item
        The Fourier sparsity $s_f$ of a function $f: \G \to \mathbb{C}$ is defined to be the number of non-zero Fourier coefficients in the Fourier expansion of $f$ (\cref{Fourier_inversion}). In this paper, by sparsity of a function, we will mean the Fourier sparsity of the function.

    \item 
        Fourier support $\supp(f)$ of a function $f: \G \to \mathbb{C}$ denotes the set $\left\{ \chi \,\mid\, \widehat{f}(\chi) \neq 0\right\}$.
\end{itemize}
\end{definition}

\begin{definition}
\label{defi:pseudo-inner product}
    \textbf{(Pseudo inner product)} For $x, r \in \G$ we denote by $*$ the following pseudo inner product. $$r * x := \biggl( \sum_{i=1}^{T} \frac{N}{p_i^{m_i}} r^{(i)} \cdot x^{(i)} \biggr) \pmod{N},$$ where $N = \LCM \{ p_1^{m_1}, \ldots, p_T^{m_T} \}$.
\end{definition}

\begin{observation}\label{obs_chi_r_x}
    $r* x \neq 0 \Rightarrow \chi_r(x)\neq 1$. 
\end{observation}

\subsection*{Fourier dimension}
The Fourier dimension of a function $f:\G \to \mathbb{C}$ is the ``rank" of the support of $f$. To define it formally we need the following definition regarding the span of a set of elements in $\G$ (or $\widehat{G}$).

\begin{definition}
    Given a subset $\Gamma = \{r_1, \dots, r_t\}\subset \G$ we denote by $S_{\Gamma}$ the set
    $$\left\{ z_1\cdot r_1 + \dots + z_t\cdot r_t \mid z_1, \dots, z_t\in \mathbb{N}\right\},$$
    where $z_i \cdot r_i$ means adding $z_i$ copies of $r_i$ and the ``+" is the group operation.
\end{definition}

It can be checked that the subset $S_{\Gamma}$ is a subgroup of $\G$ and hence a normal subgroup of $\G$ as $\G$ is an Abelian group. 

\begin{definition}\label{lemma_S_gamma}
    Given a subset $\widehat\Gamma = \{\widehat{r_1}, \dots, \widehat{r_t}\}\subset \widehat{\G}$, let 
    $\Gamma = \{r_1, \dots, r_t\}\subset \G$ denote the set of corresponding elements of $\G$ (according to Equation~\ref{eq:character}). That is $\chi_{r_i} = \widehat{r_i}$. Thus, the set $S_{\widehat{\Gamma}}$ is the set $\left\{ \chi_r \mid r\in S_{\Gamma}\right\}.$ In other words, the set $S_{\Gamma}\subset \G$ and the set $S_{\widehat{\Gamma}}\subset \widehat{\G}$ are the homomorphic map of each other.  By abuse of notation, we use $S_{\Gamma}$ to denote $S_{\widehat{\Gamma}}$ also. 
\end{definition}

Finally we can formally define the 
Fourier dimension of a Boolean function $f:\mathcal{G} \to \pmo$.

\begin{definition}\label{defn_Fourier_dimension}
    Let $f:\G \to \pmo$ be a Boolean valued function. The Fourier dimension of $f$, denoted by $r(f)$, is defined as the minumum $t$ such that there exists $\Gamma \subset \G$ with $|\Gamma| = t$ and $\supp(\widehat{f}) \subseteq S_{\Gamma}$. 
\end{definition}

\subsection*{Partitioning of \texorpdfstring{$\G$}{G} into ``linear subspaces"}

The group $\mathbb{Z}_2^n$ is also a vector space over the field $\mathbb{F}_2$, and hence has a linear algebraic structure. The characters are in one-to-one correspondence with $\mathbb{F}_2$-linear forms. However, a general Abelian group is not a vector space and does not enjoy these favorable properties. But one can define something along similar lines. 

For ease of presentation we will assume that $\G = \mathbb{Z}_{p_1^{m_1}} \times \cdots \times \mathbb{Z}_{p_T^{m_T}}$, where $p_1^{m_1}, \ldots, p_T^{m_T}$ may not be distinct. Thus an element $r$ of $\G$ is 
of the form $(r^{(1)}, \dots, r^{(T)})$, where $r^{(i)}$ is an element of $\mathbb{Z}_{p_i^{m_i}}$.

For $x, r \in \G$ we denote by $*$ the following. $$r * x := \biggl( \sum_{i=1}^{T} \frac{\lcmG}{p_i^{m_i}} r^{(i)} \cdot x^{(i)} \biggr) (\mod \lcmG),$$ where $\lcmG = \LCM \{ p_1^{m_1}, \ldots, p_T^{m_T} \}$.


We define something similar to a linear space in spirit.

\begin{definition}\label{defn_normal_subgroup}
\label{defn_cosets_groups}
    Let $r_1,\ldots,r_t \in \G$ such that $r_j= (r_j^{(1)},\ldots, r_j^{(T)})$, where $r_{j}^{(i)} \in \mathbb{Z}_{p_i^{m_i}}$ for all $i \in [T]$ and $j \in [t]$. Let $b= (b_1, \ldots, b_t) \in \mathbb{Z}_\lcmG^t$, where $\lcmG = \LCM\{ p_1^{m_1}, \ldots, p_T^{m_T}\}$. We define the set $V_{b,r_1,\ldots,r_t}$ by
    \begin{align}\label{eqn_cosets_groups}
        V_{b,r_1,\ldots,r_t} = \{ x\in \G : r_j * x = b_j (\mod \lcmG) \ \forall j \in [t]\}.
    \end{align}
     That is, 
    \begin{align*}
        V_{b,r_1,\ldots,r_t} 
        &= \biggl\{x \in \G : \biggl( \sum_{i=1}^{T} \frac{\lcmG}{p_i^{m_i}} r_j^{(i)} \cdot x^{(i)} \biggr) = b_j (\mod \lcmG) \ \forall j \in [t] \biggr\}.
    \end{align*}

    If $b_1 = b_2 = \dots = 0$ then we use $V_{0,r_1,\ldots,r_t}$ to imply $V_{b,r_1,\ldots,r_t}$.
\end{definition}



\begin{lemma}\label{lemma_cosets_of_V_0}
    $V_{0,r_1,\ldots,r_t}$ is a normal subgroup of $\G$ and for any $b, r_1, \dots, r_t$ either $V_{b,r_1,\ldots,r_t}$ is a coset of $V_{0,r_1,\ldots,r_t}$ or $V_{b,r_1,\ldots,r_t} = \emptyset$.
\end{lemma}

\begin{proof}
    Clearly, $0\in V_{0,r_1,\ldots,r_t}$, where $0$ is the identity element of $\G$. Let $x,y \in \G$. Then, since $r_j^{(i)} * (x-y)^{(i)} = \sum_{i=1}^{T} \frac{\lcmG}{p_i^{m_i}} r_j^{(i)} \cdot x^{(i)} - \sum_{i=1}^{T} \frac{\lcmG}{p_i^{m_i}} r_j^{(i)} \cdot y^{(i)} =0 (\mod \lcmG)$ for all $i\in [T], j \in [t]$, so $x-y \in V_{0,r_1,\ldots,r_t}$. So $V_{0,r_1,\ldots,r_t}$ is a subgroup of $\G$, and hence a normal subgroup of $\G$ since $\G$ is Abelian.

    Since $V_{0,r_1,\ldots,r_t}$ is a normal subgroup of $\G$ we can now consider its cosets. 
    For each coset $B$ of $V_{0,r_1, \ldots,r_t}$, let $a(B)$ be a fixed coset representative (we choose a coset representative and fix it all through the proof). Then, if $y \in B$, then $y = a(B)+x$, where $x \in V_{0,r_1,\ldots,r_t}$. Now, for each $r_j$, 
    $$r_j * y = r_j * (a(B)+x) = r_j * a(B),$$ 
    since $r_j * x =0$ for all $r_j, \ j \in [t]$. So, $r_j * y$ is fixed for each coset $B$. If we assume $r_j * a(B) = b_j (\mod \lcmG)$ for each $j \in [t]$, then we have $V_{b,r_1,\ldots,r_t}$ as a coset of $V_{0,r_1,\ldots,r_t}$, where $b=(b_1, \ldots b_t)$. 

    If for a $b = (b_1, \ldots,b_t) \in \mathbb{Z}_\lcmG^t$, there does not exist any $a(B)$ such that $r_j * a(B) = b_j \ \forall j \in [t]$, $V_{b, r_1, \ldots, r_t} = \emptyset$. 
\end{proof}

Since $V_{0,r_1,\ldots,r_t}$ is a normal subgroup by \cref{lemma_cosets_of_V_0}, we can partition $\G$ with the cosets of $V_{0,r_1,\ldots,r_t}$. 




The co-dimension of a coset $V_{b,r_1,\ldots,r_t}$ is defined as

\begin{definition}\label{defn_codimension}
    \textbf{(Codimension)} Let $V_{b,r_1,\ldots,r_t}$ be as defined in \cref{defn_cosets_groups}. Then the codimension of $V_{b,r_1,\ldots,r_t}$ is given by 
    \begin{align*}
        \Codim(&V_{b,r_1,\ldots,r_t}) =  \min_{t'} \biggl\{t' : \exists r_1',\ldots, r_{t'}' \in \{r_1, \ldots, r_t\} \text{ such that } \\
        &V_{b,r_1,\ldots,r_t} = \{x\in \mathcal{G}: r_j' * x =b_j \pmod{\lcmG} \quad \forall j \in [t']\} \biggr\},
    \end{align*}
    where $b =(b_1, \ldots, b_t) \in \mathbb{Z}_\lcmG^t$, and $\lcmG = \LCM\{p_1^{m_1}, \ldots, p_T^{m_T}\}$. 
\end{definition}

\subsection*{Restriction of a function on a coset}


Let $f:\G \to \{\-1,+1\}$ be a function and let $\Gamma=\{r_1,\ldots,r_t\}$ be a set of elements in $\G$.
Let $\mathcal{C}_{r_1, \ldots, r_t}(f)$ be the set of cosets of $S_{\Gamma}$, which partitions the Fourier support of $f$. That is, 
\begin{align}\label{eqn_cosets_of_S_gamma}
    \mathcal{C}_{r_1, \ldots, r_t}(f) = \{r+ S_\Gamma | (r+S_\Gamma) \cap \supp(\widehat{f}) \neq \phi, r \in \G\}.
\end{align}

Given a set of elements $\{r_1, \ldots, r_t\} \subseteq \G$ and $b=(b_1, \ldots, b_t) \in \mathbb{Z}_n^t$, where $n=\LCM\{p_1^{n_1}, \ldots, p_T^{n_T}\}$, we define the restriction of 
function  $f:\G \to \pmo$ on the coset $V_{b,r_1,\ldots,r_t}$, as
the usual restriction of the function $f$ on $V_{b,r_1,\ldots,r_t}$ which is a subset of the domain $\G$. That is, $f|_{V_{b,r_1,\ldots,r_t}}$ is a function from $V_{b,r_1,\ldots,r_t}$ to $\pmo$. The following lemma helps us to understand the Fourier coefficients of $f|_{V_{b,r_1,\ldots,r_t}}$.

\begin{lemma}\label{lemma_restriction}
    $f|_{V_{b,r_1,\ldots,r_t}}$ can be uniquely represented as a linear combination of the characters of $\G$.
\end{lemma}

\begin{proof}
    Let $\mathcal{C}_{r_1,\ldots,r_t}(f)$ be the set defined in \cref{eqn_cosets_of_S_gamma}.
    We can write the Fourier expansion of $f$ in the following way.
    \begin{align*}
        f(x) &= \sum_{C \in \mathcal{C}_{r_1,\ldots,r_t}(f)} \sum_{r+r(C) \in r(C)+S_\Gamma} \widehat{f}(r(C)+r) \chi_{r(C)+r}(x) \\
        &= \sum_{C \in \mathcal{C}_{r_1,\ldots,r_t}(f)} \bigg ( \sum_{r\in S_\Gamma} \widehat{f}(r(C)+r) \chi_r(x)  \bigg ) \chi_{r(C)}(x) \\
        &= \sum_{C \in \mathcal{C}_{r_1,\ldots,r_t}(f)} \bigg ( P_C(\chi_{r_1}(x), \ldots, \chi_{r_t}(x)) \bigg ) \chi_{r(C)}(x),
    \end{align*}
    where $r(C)$ is the representative of the coset $C\in \mathcal{C}_{r_1,\ldots,r_t}(f)$ and $$P_C(\chi_{r_1}(x), \ldots, \chi_{r_t}(x))=\sum_{r\in S_\Gamma} \widehat{f}(r(C)+r) \chi_r(x).$$ If we fix $r_j * x = b_j (\mod n)$ for all $j\in [k]$, where $x \in \G$, then $f$ becomes $f|_{V_{b,r_1,\ldots,r_t}}$, which follows from the definition of $V_{b,r_1,\ldots,r_t}$ and $f|_{V_{b,r_1,\ldots,r_t}}$. Its characters are $\chi_{r(C)}, \ C \in \mathcal{C}_{r_1,\ldots,r_t}(f)$, and the coefficients are given by
    \begin{align*}
        \widehat{f|_{V_{b,r_1,\ldots,r_t}}} (\chi_{r(C)}) &= P_C(b_1, \ldots, b_t)
        = \sum_{r\in S_\Gamma} \widehat{f}(r(C)+r) \chi_r(x).
    \end{align*}
\end{proof}

The Fourier sparsity and dimension of $f|_{V_{b,r_1,\ldots,r_t}}$ are defined in the same way as  in the case of $f$ in \cref{def:fourier-sparsity} and \cref{defn_Fourier_dimension} respectively.
Now, let us look at the following lemma, which plays an important role in understanding how to reduce the sparsity of a function $f: \G \to \mathbb{C}$ by partitioning the group $\G$ into cosets.



\subsection{The subgroup $H^\perp$}\label{section_H_perp}
\begin{definition}\label{defn_1_new}
    \textbf{(The subgroup $H^\perp$)} Let $H$ be a subgroup of $\G$. Then $H^\perp$ is the subgroup given by
    \begin{align*}
        H^\perp := \{ x \in \G : x * y =0 \ \forall y \in H \},
    \end{align*}
    where $$x * y := \biggl( \sum_{i=1}^{T} \frac{N}{p_i^{m_i}} x^{(i)} \cdot y^{(i)} \biggr) \pmod{N},$$ and $N = \LCM \{p_1^{m_1} \cdots p_T^{m_T}\}$.
\end{definition}

\begin{remark}
    Observe that $H^\perp$ is also a subgroup of $\G$.
\end{remark}

\begin{claim}\label{claim_H_perp_new}
    For $z \in H^\perp$, 
    \begin{align*}
        \sum_{\beta \in H} \omega_N^{\beta * z} =|H|,
    \end{align*}
    where $\omega_N$ is a primitive $N^{th}$ root of unity of order $N$, and $|H|$ denotes the order of the subgroup $H$. 
    
    Also, for $z \notin H^\perp$, 
    \begin{align*}
        \sum_{\beta \in H} \omega_N^{\beta * z} =0.
    \end{align*}
\end{claim}

\begin{proof}
    \begin{description}
        \item [Case 1: When $z\in H^\perp$.] Then $z * x =0$ for all $x \in H$. Hence 
    \begin{align*}
        \sum_{\beta \in H} \omega_N^{\beta * z} =|H|.
    \end{align*}

        \item [Case 2: When $z\notin H^\perp$.] Let $\sum_{\beta \in H} \chi_\beta (z) =A$. Since $z \notin H^\perp$, so, by Definition~\ref{defn_1_new}, there exists $\gamma \in H$ such that $\gamma * z \neq 0$, that is, $\chi_\gamma(z) \neq 1$. (see Observation~\ref{obs_chi_r_x})
        Then, 
        \begin{align*}
            \chi_\gamma(z) \times A &= \chi_\gamma(z) \sum_{\beta \in H} \chi_\beta(z) \\
            &= \sum_{\beta \in H} \chi_{\beta+\gamma}(z) \\
            &= \sum_{\gamma' \in H} \chi_{\gamma'}(z), &\gamma' = \beta+\gamma \\
            &= A,
        \end{align*}
        which implies that 
        \begin{align*}
            A (\chi_\gamma(z)-1) =0 \Rightarrow A=0,
        \end{align*}
        since $\chi_\gamma(z) \neq 1$.
    \end{description}
\end{proof}

We need the following isomorphism result:

\begin{lemma}\label{lemma_iso_H_perp_G/H}
    $H^\perp$ is isomorphic to the quotient group $\G/H$.
\end{lemma}

\begin{proof}
    We will show that $\widehat{H^\perp}$ is isomorphic to $\widehat{\G/H}$, which implies that $H^\perp$ is isomorphic to the quotient group $\G/H$, as $G \equiv \widehat{G}$ for any group $G$.

    The set of characters of $H^\perp$ is given by $$\Ann_{\G}(H) = \{\chi \in \widehat{\G} : \chi(y) = 1 \ \forall y \in H\},$$ where $\Ann_{\G}(H)$ is known as the annihilator of the subgroup $H$ in $\G$. From Definition~\ref{defn_1_new}, observe that for $x \in \G$, $\chi_x(y) = \omega_\mathcal{L}^{x*y} = 1$  if and only if $x* y =0 \Leftrightarrow x \in H^\perp$.

    Let us define a group homomorphism $\mathcal{F} : \widehat{\G/H} \to \Ann_{\G}(H)$ by $\mathcal{F}(\zeta) = \zeta \circ q$, where $q: \G \to \G/H$ is the quotient is the quotient group homomorphism defined by $q(r)=r+H$. 
    \begin{itemize}
        \item \textbf{($\mathcal{F}$ is injective)} Let $\zeta \in \widehat{\G/H}$ such that $\zeta \circ q = \Tilde{0}$, where $\Tilde{0}=(0, \ldots, 0) [T \text{ times}]$ is the identity element of $\G$. So, $\zeta(r + H) = \zeta \circ q(r) = 1$ for all $r \in \G$, which implies $\zeta$ is the identity element of $\G/H$. Therefore, $\mathcal{F}$ is injective.

        \item \textbf{($\mathcal{F}$ is surjective)} Let $\psi \in \Ann_{\G}(H)$. Let $\zeta : \G/H \to \mathbb{C}$ by $\zeta(r+H) = \psi(r)$ for all $r \in \G$. Since $\chi_{r+H}(x) = \omega_\mathcal{L}^{r *x + H*x}$, so any character $\chi_{r+H}$ of $\G$ is a character of $\G/H$ if $H * x=0$ for all $x\in \G$ (as then, the value of $\omega_\mathcal{L}^{r *x + H*x}$ will be determined only by the coset representatives). Since $\psi(r+H) = \psi(r)\psi(H) = \psi(r)$ and $H$ is the identity element of $G/H$, so $\zeta$ is a character of $\G/H$. Also, $\psi = \zeta \circ q$. Therefore, $\mathcal{F}(\zeta) = \zeta \circ q =\psi$. Hence, $\mathcal{F}$ is surjective.
    \end{itemize}

    Therefore, $\mathcal{F}$ is an isomorphism, which implies that $H^\perp$ is isomorphic to the quotient group $\G/H$.
\end{proof}

\begin{corollary}\label{lemma_H_perp}
    $|H| \times |H^\perp| = |\G|$.
\end{corollary}

\begin{proof}
    Follows from the fact that $|H^\perp|= \frac{|\G|}{|H|}$, since $H^\perp$ is isomorphic to the quotient group $\G/H$ by Lemma~\ref{lemma_iso_H_perp_G/H}, and $|\G/H| = \frac{|\G|}{|H|}$ by Lagrange's theorem.
\end{proof}

\subsection{Basics from Probability}

We will need some notions and lemmas from probability. We will need the following notion of covariance between two random variables. 

\begin{definition}[Covariance]\label{defn_covariance}
    Let $X_1$ and $X_2$ be two random variables. Then, the covariance of $X_1$ and $X_2$ is given by $$\cov(X_1,X_2) = \mathbb{E}[X_1X_2] - \mathbb{E}[X_1] \mathbb{E}[X_2].$$
\end{definition}

We will need the following concentration lemmas. Let us look at Markov's inequality first.

\begin{lemma}[Markov's inequality]
    Let $X$ be a nonnegative random variable and $a>0$, then,
    $$\Pr[X \geq a] \leq \frac{\mathbb{E}[X]}{a},$$
    where $\mathbb{E}[X]$ is the expected value of $X$.
\end{lemma}

Let us now state Chernoff's bound for random variables taking value in [-1, 1].

\begin{lemma}\label{Chernoff}
\textbf{(Chernoff bound)} Let $X_1, \ldots, X_n$ be independent random variables taking values in $\{0,1\}$. Let $X= \sum_{i=1}^n X_i$ and $\mathbb{E}[X]$ be the expected value of $X$. Then,
    \begin{enumerate}
        \item $\Pr[X \leq (1-\delta) \mathbb{E}[X]] \leq e^{-\frac{\delta^2\mathbb{E}[X]}{2}}$, $0<\delta<1$.
        \item $\Pr[X \geq (1+\delta) \mathbb{E}[X]] \leq e^{-\frac{\delta^2\mathbb{E}[X]}{2+\delta}}$, $\delta\geq 0$.
    \end{enumerate}
\end{lemma}

Now, we state the Hoeffding's inequality for random variables taking value in [-1, 1].

\begin{lemma}[Hoeffding's Inequality] \label{lem:chernoff-sample}
Let $X_1, \cdots, X_k$ be real independent random variables, each taking value in [-1, 1]. Then 
\[\Pr \left[\left|\sum_{i=1}^k X_i - \mathbb{E}\sum_{i=1}^k X_i\right| \geq \epsilon\right] \le 2\exp(-\frac{\epsilon^2}{2k}) \]
\end{lemma}


\subsection{Basics from algebraic number theory} \label{sec:nt}
Fix $n \in \N$, and let us denote $e^{2\pi\iota/n}$, a primitive $n$-th root of unity, by $\omegan$. Let $\Q(\omegan)$, be the field obtained by adjoining $\omegan$ to $\Q$, the field of rationals. One can think of $\Q(\omegan) = \{r(\omegan) \mid r \in \Q[x]\}$, where $\Q[x]$ is the ring of polynomials with coefficients from $\Q$. Similarly, one can define $\Z[\omegan]$, by adjoining $\omegan$ to the domain of integers $\Z$.

\begin{definition}{\bf (Cyclotomic Integer)}
Any $\alpha \in \C$ is called a cyclotomic integer, if $\alpha \in \Z[\omegan]$.
\end{definition}
Recall that $\alpha \in \C$ is an {\em algebraic number} if it is the root of a non-zero polynomial in $\Q[x]$. The {\em minimal polynomial} of $\alpha$ is the {\em unique} monic polynomial in $\Q[x]$ having $\alpha$ as a root. If the minimal polynomial has {\em integer coefficients}, then we say that $\alpha$ is an {\em algebraic integer}. It is a well-known fact that a cyclotomic integer is an algebraic integer. For details, see~\cite[Chapter 13]{ireland1990classical}, or expository lecture notes~\cite[Chapter 2]{green20} \& \cite{sauder13}.

Recall that the {\em Galois} group, denoted $\Gal(\Q(\omegan)/\Q)$ of automorphisms of $Q(\omegan)$; it is {\em isomorphic} to the {\em multiplicative group} $\Z_n^{\ast}$, of integers $\bmod~n$, i.e.,~all integers $k$ such that gcd$(k,n)=1$. Note that $|\Z_n^{\ast}| = \varphi(n)$, where $\varphi$ is the Euler's totient function, where $\varphi(n) := n \cdot \prod_{\text{prime}~p \mid n} (1-1/p)$. In simple words, think of $\Gal(\Q(\omegan)/\Q)$, as consisting $\varphi(n)$ many automorphism $\sigma_k$, that sends $\omegan \mapsto \omegan^{k}$, where $k \in \Z_n^{\ast}$. For basics, we refer to~\cite[Chapter 1]{green20}.

Finally, as follows, one can define {\em norm} of an $\alpha \in \Q(\omegan)$.
\begin{definition}[Norm]
The norm of $\alpha \in \Gal(\Q(\omegan)/\Q)$ is defined by
\[
N_{\Gal(\Q(\omegan)/\Q}(\alpha)\;:=\;\prod_{\sigma \in \Gal(\Q(\omegan)/\Q)}\,\sigma(\alpha)\;.
\]
\end{definition}
It is a well-known fact that for a non-zero $\alpha$, $N_{\Gal(\Q(\omegan)/\Q}(\alpha) \in \Q-\{0\}$; see~\cite[Lemma 1.6.1]{green20}. In particular, when $\alpha$ is a cyclotomic integer, i.e,~$\alpha \in \Z[\omegan]$, it is well-known that $N_{\Gal(\Q(\omegan)/\Q}(\alpha) \in \Z$, see~\cite[Corollary~2.2.2]{green20}. We restate the same as a theorem below in the most simplified way without the above terminologies since this fact has been crucially used throughout the paper.

\begin{theorem}
For $n \in \Z$, let $\omegan$ be a primitive root of unity. Let $f\in \mathbb{Z}[x]$, such that $f(\omega_n)\neq 0$. Then,
\[
|\prod_{i \in \Z_n^{\ast}} f(\omegan^i)|\;\ge\;1\;. 
\]
\end{theorem}
In retrospect, one can think of $\alpha = f(\omegan) \in \Z[\omegan]$. Note that, $\sigma_k(\alpha) = f(\omegan^k)$, by the definition of the automorphism $\sigma: \omegan \mapsto \omegan^k$, where gcd$(k,n)=1$. Therefore, their product, which is the norm $N_{\Gal(\Q(\omegan)/\Q)}(\alpha)$, must be a {\em non-zero integer}, and hence the conclusion follows. Using Theorem~\ref{theorem_roots_bound}, one can directly conclude the following fact (which can also be proved in an elementary way).
\begin{fact}\label{fact:non-zero-anypower}
Let $f \in \Z[x]$, such that $f(\omega_{N}) \ne 0$. Then, $f(\omega_{N}^i) \ne 0$, for any $i$ such that $\gcd\{i,N\}=1$.
\end{fact}

\subsection{Miscellaneous}

Recall the cosets $V_{b, r_1, \ldots, r_k}$ from \cref{defn_cosets_groups}.

\begin{theorem}\label{thm:random-coset}
Let $S\subseteq \G$ be such that $|S| \leq s+1$. Then, if $k \geq 2\log_\lcmG s + \log_\lcmG (\frac{1}{\delta})$, all $\chi \in S$ belong to different cosets except with probability at most $\delta$. 
\end{theorem}

\begin{proof}
    Let $\chi_{S_1}$ and $\chi_{S_2}$ be two distinct characters in $\widehat{\G}$. We have that 
    \begin{align*}
        \Pr \left[ \chi_{S_1}, \chi_{S_2} \text{ belong to the same coset} \right]
        =\Pr \left[ \forall i \in [k], r_i * (S_1-S_2) =0 \right]
        =\frac{1}{\lcmG^k},
    \end{align*}
    where the last equality holds because $\chi_{S_1}$ and $\chi_{S_2}$ are distinct, and $r_1, \ldots, r_k$ are independent, and $\lcmG = \LCM\{p_1^{m_1}, \ldots, p_T^{m_T}\}$.
    
     Since $|S|\leq s+1$, therefore the number of ways in which two distinct $\chi_{S_1}, \chi_{S_2}$ can be chosen from $S$ is $\binom{s+1}{2} \leq s^2$. Therefore, probability that all $\chi\in S$ belong to different cosets is given by
        \begin{align*}
            \Pr[\text{all } \chi\in S \text{ belong to different cosets}] &\geq 1- \binom{s+1}{2} \frac{1}{\lcmG^k} \\
            &\geq 1- s^2 \frac{1}{\lcmG^k} \\
            &\geq 1- s^2 \frac{1}{\lcmG^{2\log_\lcmG s + \log_\lcmG (\frac{1}{\delta})}} \\
            &= 1- s^2 \frac{1}{\frac{s^2}{\delta}} \\
            &= 1- \delta.
        \end{align*}
\end{proof}

\begin{definition}\label{defn_kronecker}
    \textbf{(Kronecker product)} Let $M_1$ be a $a_1\times b_1$ matrix and $M_2$ be a $a_2\times b_2$ matrix. Then the Kronecker product $M_1 \otimes M_2$ of $M_1$ and $M_2$ is a $a_1a_2 \times b_1b_2$ block matrix, and is given by

    \begin{align*}
        M_1 \otimes M_2 = 
        \begin{bmatrix}
            m_{11} M_2 & \cdots & m_{1b_1} M_2 \\
            \cdot & \cdot & \cdot \\
            \cdot & \cdot & \cdot \\
            \cdot & \cdot & \cdot \\
            m_{a_1 1} M_2 & \cdots & m_{a_1b_1} M_2
        \end{bmatrix}
    \end{align*}
\end{definition}

\begin{definition}\label{defn_random_subgroup}
    \textbf{(Random subgroup)} A random subgroup of $\G$ is a subgroup of the form $V_{0,r_1, \ldots, r_t}$ (see \cref{defn_cosets_groups} and \cref{lemma_cosets_of_V_0}), where $r_1, \cdots, r_t$ are chosen independently and uniformly at random from $\G$.
\end{definition}

\begin{definition}\label{defn_t_dim_coset}
    \textbf{($t$-dimensional coset structure)} Let $H$ be a subgroup of $\G$ which is of the form $V_{0,r_1, \ldots, r_t}$ (see \cref{defn_cosets_groups}), where $r_1, \cdots, r_t$ are chosen independently and uniformly at random from $\G$. Also let $\mathcal{C}$ be the set containing all the cosets of $V_{0,r_1, \ldots, r_t}$, (for details see \cref{lemma_cosets_of_V_0}). Then $(H,\mathcal{C})$ forms a random $t$-dimensional coset structure.
\end{definition}

\color{black}

Let us define the $\ell_2$ distance between two functions below. 

\begin{definition}[$\ell_2$ distance]\label{defn_ell_2_dist}

    Let $f$ and $g$ be two functions with domain $\G$ and range $\mathbb{C}$. Then the square of the $\ell_2$ distance between $f$ and $g$ is defined as $\mathbb{E}_{x\in \G} [|f(x)-g(x)|^2].$ 
    By Parseval's identity the square of the $\ell_2$-distance between $f$ and $g$ can also be written as 
    $\sum_{\chi \in \widehat{\mathbb{Z}_p^n}} |\widehat{(f-g)}(\chi)|^2$.

\end{definition}

\begin{lemma}\label{lemma_bound_H}
    Let $f,g$ be two Boolean valued functions from $\G$ to $\{-1,+1\}$. Then, $$|\widehat{fg} (\chi)| \leq ||f||_2 ||g||_2,$$ for any character $\chi \in \widehat{\G}$.
\end{lemma}

\begin{proof}
    \begin{align*}
        |\widehat{fg} (\chi)| &= |\sum_{\psi \in \widehat{\G}} \widehat{f}(\psi) \widehat{g}(\chi+\psi)| \\
        &\leq \sqrt{\sum_{\psi \in \widehat{\G}} \widehat{f}(\psi)^2} \sqrt{\sum_{\psi \in \widehat{\G}} \widehat{g}(\chi+\psi)^2} \text{ by Cauchy-Schwartz inequality} \\
        &= ||f||_2 ||g||_2 \text{ by } \cref{theorem_Parseval}.
    \end{align*}
\end{proof}

Now let us formally define the notion of $\epsilon$-close and $\epsilon$-far in $\ell_2$ below. 

\begin{definition}[$\mu$-close to $s$-sparse]\label{defn_B_S}

    Let $f$ and $g$ be two functions with domain $\G$ and range $\mathbb{C}$. Then the square of the $\ell_2$ distance between $f$ and $g$ is defined as $\mathbb{E}_{x\in \G} [|f(x)-g(x)|^2].$ 
    By Parseval's identity the square of the $\ell_2$-distance between $f$ and $g$ can also be written as 
    $\sum_{\chi \in \widehat{\G}} |\widehat{(f-g)}(\chi)|^2$.

    We say that $f$ is $\epsilon$-close to $g$ in $\ell_2$ if the square of the $\ell_2$ distance between $f$ and $g$ is less than $\epsilon$. Similarly, $f$ is $\epsilon$-far from $g$ in $\ell_2$ if the square of the $\ell_2$ distance between $f$ and $g$ is at least $\epsilon$.
\end{definition}


The following lemma gives us a relation between the $\ell_2$ distance between two Boolean valued functions $f$ and $g$ defined in \cref{defn_B_S} and $\Pr_x[x\in \G | f(x)\neq g(x)]$.

\begin{lemma}\label{lemma_Boolean_distance_equivalence}
    The square of the $\ell_2$ distance between two Boolean valued functions $f$ and $g$ defined in \cref{defn_B_S} is equal to $4\Pr_x[x\in \G | f(x)\neq g(x)]$. \footnote{If $f$ and $g$ are two Booelan-valued functions then $\Pr_x[x\in \G | f(x)\neq g(x)]$ is also called the Hamming distance between the two functions. So the $\ell_2$ norm between two Boolean-valued functions is 4 times the Hamming distance between two Boolean-valued functions.}
\end{lemma}

\begin{proof}
By Parseval (\cref{theorem_Parseval}), we have,
    \begin{align*}
        \mathbb{E}_{x\in \G} [(f-g)(x)\overline{(f-g)(x)}] &= \sum_{\chi\in \widehat{\G}} |\widehat{(f-g)}(\chi)|^2.
    \end{align*}

    Also, 
    \begin{align*}
        \mathbb{E}_{x\in \G} [(f-g)(x)\overline{(f-g)(x)}] &= \frac{1}{|\G|} \times (4 \times |\{x \in \G | (f-g)(x)\neq 0\}|) \\
        &=\frac{1}{|\G|} \times (4 \times |\{x \in \G | f(x)\neq g(x)\}|) \\
        &=4 \times \Pr_x [f(x) \neq g(x)],
    \end{align*}
    where $|S|$ denotes the number of elements in the set $S$. 

    So, the square of the $\ell_2$ distance between $f$ and $g$ is equal to $4$ times $\Pr_x[x\in \G | f(x)\neq g(x)]$.
\end{proof}

Now let us define the total variation distance between two probability distributions.

\begin{definition}\label{defn_total_variation}
    Let $(\Omega, \mathcal{F})$ be a probability space, and $P$ and $Q$ be probability distributions defined on $(\Omega, \mathcal{F})$. The total variation distance between $P$ and $Q$ is defined in the following way.
    \begin{align*}
        ||P-Q||_{TV} = \sup_{A \in \mathcal{F}} |P(A) - Q(A)|.
    \end{align*}
\end{definition}

\begin{lemma}\label{lemma_total_variation}
    Given two probability distributions $P$ and $Q$ on a probability space $(\Omega, \mathcal{F})$, the total variation distance between $P$ and $Q$ is half of the $L_1$ distance between them. That is,
    \begin{align*}
        ||P-Q||_{TV} = \frac{1}{2} \sum_x |P(x) - Q(x)|.
    \end{align*}
\end{lemma}






\section{Proofs of structure Theorems~\texorpdfstring{\ref{lem:structure}}{} and ~\texorpdfstring{\ref{structure_theorem_2}}{}} \label{sec:lb-z_p}

In this section, we will prove \cref{lem:structure} and \cref{structure_theorem_2} which are generalizations of Gopalan et al's result ~\cite{gopalan2011testing} to the case when the domain of the function is $\G$. Note that Chakraborty et al~\cite{DBLP:conf/mfcs/0001DDGS24} proved these results for groups of the form $\mathbb{Z}_{p_1}^{n_1} \times \cdots \times \mathbb{Z}_{p_T}^{n_T}$, where $p_i, i \in [T]$ are distinct primes. Since $\G$ is not a vector space in general, we use the pseudo-inner product, normal subgroups, and cosets (see Section~\ref{section_Preliminaries} for more details) to prove our results. Also, let $\omega_L$ be a primitive $L^{th}$ root of unity, $L$ is a positive integer. For any polynomial $\sum_{i=0}^d g_ix^i =g(x) \in \C[x]$, let $|g|_1\;:=\;\sum_{i=0}^d |g_i|$.


\subsection{The projection on a coset}
To prove Theorem~\ref{lem:structure}, we need the following. Here $\G$ denotes a finite Abelian group.

\begin{definition}
\label{def:proj-op}
\textbf{(Projection  $P_{r+H}$)}
Given a subgroup $H \subseteq \G$ and a coset $r+H$, let us define the projection operator $P_{r+H}$ on functions $f:\G \to \mathbb{R}$ as follows:
\begin{align*}
        \widehat{P_{r+H}f} (\chi) = 
        \begin{cases}
            \widehat{f}(\chi)  & \chi \in r+H \\
            0  & \text{otherwise}.
        \end{cases}
    \end{align*}
    So, $$P_{r+H}f(x)=\sum_{\chi\in r+H} \widehat{f}(\chi) \chi(x).$$
\end{definition}

To prove \cref{granularity_of_Z_p}, we need to prove Theorem~\ref{lem:structure} and 
to prove Theorem~\ref{lem:structure}, we need the following lemma.

\begin{lemma}\label{lemma_granularity_Z_p}
Given a subgroup $H$ of $\G$, we have
$$
    P_{r+H} f(x)\;=\; \E_{z\in H^\perp} [f(x-z)\chi_r(z)] \;=\; \frac{1}{|H^\perp|} \sum_{z\in H^\perp} f(x-z) \chi_r(z)\;.
$$
\end{lemma}

\begin{proof}[Proof of \cref{lemma_granularity_Z_p}]
From the definition of the projection operator $P_{r+H}$ (\cref{def:proj-op}) we get 
\begin{proof}
\begin{align*}
    \mathcal{P}_{r+H} f(x)
    &= \sum_{\gamma \in r+H} \hat{f}(\gamma) \chi_\gamma(x) & \\
    &= \sum_{\beta \in H} \hat{f}(r+\beta) \chi_{r+\beta} (x) &\text{ where } \gamma = r+\beta, \beta \in H \\
    &= \sum_{\beta \in H} \hat{f}(r+\beta) \chi_r(x) \chi_\beta(x) \\
    &= \chi_r(x) \sum_{\beta \in H} \hat{f}(r+\beta) \chi_\beta(x) \\
    &= \chi_r(x) \sum_{\beta\in H} \chi_\beta (x) \frac{1}{|\G|} \sum_{y\in \G} f(y) \overline{\chi_{r+\beta} (y)} \\
    &= \frac{1}{|\G|} \chi_r(x) \sum_{y\in \G} f(y) \sum_{\beta \in H} \chi_\beta(x) \chi_{-r-\beta} (y) \\
    &= \frac{1}{|\G|} \chi_r(x) \sum_{y\in \G} f(y) \sum_{\beta \in H} \chi_\beta(x) \chi_{r+\beta} (-y) \\
    &= \frac{1}{|\G|} \chi_r(x) \sum_{y\in \G} f(y) \chi_r(-y) \sum_{\beta \in H} \chi_\beta (x-y) \\
    &= \frac{|H|}{|\G|} \sum_{y \in H^\perp} f(y) \chi_r (x-y) &\text{by } \cref{claim_H_perp_new} \\
    &= \frac{1}{|H^\perp|} \sum_{z \in H^\perp} f(x-z) \chi_r (z) &\text{by Lemma}~\ref{lemma_H_perp} \\
    &= \mathbb{E}_{z \in H^\perp} [f(x-z)\chi_r(z)],
\end{align*}
where $z = x-y$.
\end{proof}
\end{proof}

\subsection{Proof of \texorpdfstring{\cref{lem:structure}}{}}\label{section_structure_theorem_1}
In this subsection we will prove \cref{lem:structure}. All  Recall the definition of granularity for complex numbers (\cref{defn_granularity}) and the definition of $\mu$-close to granular (\cref{defn_mu_close_to_granular}).

\begin{proof}[Proof of Theorem~\ref{lem:structure}]
Set $k= \lceil \log_{\lcmG} 2s \rceil$,
and let $A$ be a $k \times t$ matrix with the entries $a_{ij}$ such that $a_{ij} \text{ are invertible in $\mathbb{Z}_{\lcmG}$ for all } i \geq j$, and $b$ be a $k\times 1$ matrix. 

Let $E$ be the event that the matrix $A$ is chosen such that its entries $a_{ij}$ are invertible in $\mathbb{Z}_{\lcmG}$ for all $i \geq j$, and $a_{ij}$ are chosen independently and uniformly at random. 

Let $A^\perp \subset \G$, be the set of solutions to $A\chi =0$. Let $H$ be the coset of $A^\perp$ that are solutions to system of linear equations $A\chi =b$. From the definition, 
\begin{equation}
    P_H f (x)\;=\; \sum_{\chi \in H} \widehat{f}(\chi) \chi(x).
\end{equation}

We need to show that for each $\chi_{r_i} \in B$, there exists a $k\times n$ matrix $A$, with the entries $a_{ij} \text{ are invertible in $\mathbb{Z}_{\lcmG}$ for all } i \geq j$, and a $k\times 1$ matrix $b$ such that  
\begin{align}\label{eqn_B}
    \{\chi \,|\, \chi\in B \} \bigcap \left\{\chi \,|\, A\chi \,=\,b \right\} \;=\; \{\chi_{r_i}\},
\end{align}

and, the $\ell_2$ Fourier weight of $S\cap H$ is upper bounded by $\frac{\mu^2}{s}$, that is,
\[
\sum_{\chi\in S\cap H} |\widehat{f}(\chi)|^2 \;\leq\; \frac{\mu^2}{s}\,,
\]
where $B$ is the set of character functions corresponding to the $s$ largest coefficients and $S$ is the rest of the characters,
and $H$ is the solution set of the system of linear equations $A\chi =b$.
Let $E_1$ and $E_2$ be the first and the second events respectively. So we need to prove that both these events occur simultaneously with non-zero probability.


To show the existence of $A$ and $b$, given $\chi_{r_i} \in B$, for which \cref{eqn_B} is true we will need to consider two cases:

\begin{description}
    \item[Case 1: $\chi_{r_i} = 0$]\footnote{Here $\chi_{r_i} = 0$ means that $\chi_{r_i} = \chi_{r}$ where $r = (0, \dots, 0) \in \G$.}. Observe that for all $A$, we have $A\chi_{r_i} = 0$, so $\Pr_A[A\chi_{r_i}=0]=1$. If entries of $A$ are sampled independently and uniformly at random from $\Z_{\lcmG}$, with $a_{ij} \text{ are invertible in $\mathbb{Z}_{\lcmG}$ for all } i \geq j$, then, for all $\chi_{r_j} \in B\setminus \{\chi_{r_i}\}$, we have 
    $$
        \Pr_{A}\left[ A\chi_{r_j} = 0 \mid A\chi_{r_i}=0, E \right] = \frac{\Pr_{A}\left[ \bigl( A\chi_{r_j} = 0, A\chi_{r_i}=0\bigr), \ E\right]}{\Pr_{A}\left[\bigl( A\chi_{r_i}=0 \bigr) ,\ E \right]} = \frac{1}{\lcmG^{k}}.
    $$

    Note that $\Pr_{A} \left[ E_1 \right]$ denotes the probability of event $E_1$ when $A$ is the random matrix generated by the above-mentioned procedure. 
    Using union bound, we get 
    $$
        \Pr_{A}\left[ \biggl( \exists j\neq i \text{ such that } A\chi_{r_j}=0 \biggr) \mid E \right] =  \frac{s-1}{2s} <1.
    $$
    
    Therefore, if $\chi_{r_i} = 0$ then there exists a $k\times t$ matrix $A$, with the entries $a_{ij}$ are invertible for all $i \geq j$ such that 
    \[
        \{\chi \,|\, \chi \in B \} \bigcap \left\{\chi \,|\, A\chi \,=\,0 \right\} = \{\chi_{r_i}\}.
    \]
    
    \item[Case 2: $\chi_{r_i} \neq 0$]\footnote{Here $\chi_{r_i} \neq 0$ means that $\chi_{r_i} = \chi_{r}$ where $r \neq (0, \dots, 0) \in \G$.}.
    Fix a nonzero element $b \in \G \setminus \{(0, \dots, 0)\}$. Again, let $A$ be a random  $k\times t$ matrix with entries of $A$ being sampled independently and uniformly at random from $\Z_{\lcmG}$ and the entries $a_{ij}$ are invertible in $\mathbb{Z}_{\lcmG}$ for all $i \geq j$. 
    If $\chi_{r_j} = \alpha \chi_{r_i}$, where $\alpha \in \Z_{\lcmG}$ and $\gcd(\alpha-1,\lcmG)=1$,  
    then 
    $$
        \Pr_{A} [A\chi_{r_j}\,=\,b\,\mid\, A\chi_{r_i}\,=\,b, \ E] =0.
    $$
    Now we will consider the case when $\chi_{r_j} \neq \alpha \chi_{r_i}$, with $\alpha \in \Z_{\lcmG}$ and $\gcd(\alpha-1,\lcmG)=1$. Observe that
    \begin{align}\label{eqn:hellworld}
        \Pr_{A} [ A\chi_{r_j}\,=\,b\,\mid\, A\chi_{r_i}\,=\,b ,\ E] &= \frac{\Pr_{A,b} [\bigl( A\chi_{r_j}\,=\,b, A\chi_{r_i}\,=\,b \bigr),\ E]}{\Pr_{A} [\bigl( A\chi_{r_i}\,=\,b \bigr) ,\ E]} \nonumber\\
        &= \frac{\Pr_{A} \left[ \biggl( \bigcap\limits_{\ell \in [k]}\left( A_{\ell}\cdot \chi_{r_j} = b_{\ell},\, A_{\ell}\cdot \chi_{r_i} = b_{\ell}\right) \biggr) ,\ E \right]}{\Pr_{A} [\bigl( A\chi_{r_i}\,=\,b \bigr) ,\ E]} \nonumber \\
        &= \frac{\prod\limits_{\ell \in [k]} \Pr_{A} \left[ \bigl( A_{\ell}\cdot \chi_{r_j}  = b_{\ell},\,  A_{\ell}\cdot \chi_{r_i} = b_{\ell} \bigr) ,\ E \right]}{\Pr_{A} [\bigl( A\chi_{r_i}\,=\,b\bigr) ,\ E]}
    \end{align}
    Note that $A_{\ell}$ and $b_{\ell}$ denote the $\ell$-th row of $A$ and $\ell$-th coordinate of $b$ respectively, and $[k]$ denotes the set $\{1, \dots, k\}$. Observe that the last equality in the above equation follows from the fact that rows of $A$ are independent and uniformly random samples of $\mathbb{Z}_\lcmG^k$. 
    Therefore 
    $$
      \Pr_{A} \left[ A_{\ell}\cdot \chi_{r_j} = b_{\ell} \ \mid \ A_{\ell}\cdot \chi_{r_i} = b_{\ell} ,\ E \right] = \frac{1}{\lcmG^2}.  
    $$

Therefore, we have,
    $$
        \Pr_{A} [ A\chi_{r_j}\,=\,b\,\mid\, A\chi_{r_i}\,=\,b ,\ E] = \frac{\prod\limits_{\ell \in [k]} \Pr_{A} \left[ A_{\ell}\cdot \chi_{r_j} = b_{\ell},\, A_{\ell}\cdot \chi_{r_i} = b_{\ell} ,\ E \right]}{\Pr_{A} [ A\chi_{r_i}\,=\,b ,\ E]} = \frac{1}{\lcmG^{k}}.
    $$

    Using the conditional variant of the union bound we get 
    $$
        \Pr_{A} [\bigl( \exists j\neq i \text{ such that } A\chi_{r_j}=b \mid A\chi_{r_i}=b \bigr) \ \mid \ E] = \frac{s-1}{2s} < 1.
    $$
\end{description}
From the above two cases we can now say that given any $\chi_{r_i} \in B$, there exists a $k\times t$ matrix $A$, with the entries $a_{ij} \text{ are invertible in $\mathbb{Z}_{\lcmG}$ for all } i \geq j$ and a $k\times 1$ matrix $b$ such that 
\[
\{\chi \,|\, \chi \in B \} \bigcap \left\{\chi \,|\, A\chi \,=\,b \right\} = \{\chi_{r_i}\}.
\]

Now, let us look at the coefficients corresponding to the characters in $S$, in order to prove that the $\ell$ Fourier weight of $S\cap H$ is upper bounded by $\frac{\mu^2}{s}$, $H$ being the solution set of the system of linear equations $A\chi =b$, with the entries $a_{ij} \text{ are invertible in $\mathbb{Z}_{\lcmG}$ for all } i \geq j$. Let us look at the expected value of $\sum_{\chi \in S\cap H} |\widehat{f}(\chi)|^2$, given $A\chi_{r_i} = b$.

\begin{align}\label{eqn_1}
    \mathbb{E}_{A,b} \biggl [  \sum_{\chi \in S\cap H} |\widehat{f}(\chi)|^2 \mid A\chi_{r_i} = b \ , \ E \biggr ] &= \sum_{\chi\in S} \Pr_{A,b} [\chi\in H \mid A\chi_{r_i} =b \ , \ E] |\widehat{f}(\chi)|^2 \nonumber \\ 
    &\leq \frac{1}{\lcmG^{k}} \times \mu^2 \leq \frac{\mu^2}{2 s}.
\end{align}

Therefore, by Markov's inequality, we have, 

\begin{align}\label{eqn_S}
    \Pr_{A,b} \bigg[\sum_{\chi \in S\cap H} |\widehat{f}(\chi)|^2 \geq \frac{\mu^2}{s} \mid A\chi_{r_i} = b \ , \ E \bigg] \leq \frac{\frac{\mu^2}{2 s}}{\frac{\mu^2}{s}} = \frac{1}{2} <1.
\end{align}

So, if we apply union bound to \cref{eqn_B} and \cref{eqn_S}, we get that the events $E_1$, $E_2$ and $E$ occur simultaneously with non-zero probability.

Therefore, we have,

$$
    P_Hf(x) = \widehat{f}(\chi_{r_i}) \chi_{r_i}(x) + \sum_{\chi\in S\cap H} \widehat{f}(\chi) \chi(x),
$$
where $\sum_{\chi\in S\cap H} |\widehat{f}(\chi)|^2 \leq \frac{\mu^2}{s}$.

Now, by \cref{lemma_granularity_Z_p}, we have, $$P_Hf(x) = \mathbb{E}_{z\in (A^\perp)^\perp} [f(x-z)\chi_r(z)],$$ which is of the form $\frac{g(\omega_{\lcmG})}{\lcmG^k}$, since by Lemma~\ref{lemma_H_perp}, $|(A^\perp)^\perp| = \frac{|\mathbb{Z}_\lcmG^T|}{|A^\perp|} = \lcmG^k$ (since we have sampled the entries of the matrix $A$ from $\mathbb{Z}_\lcmG$), where $g(X) \in \mathbb{Z}[X]$ is a nonzero polynomial with degree $\leq \lcmG-1$. So $P_Hf(x)$ is $k$-granular.

Also, let us look at the function $G_1(x)=\sum_{\chi\in S\cap H} \widehat{f}(\chi) \chi(x)$. We know, by \cref{eqn_1}, that, $\mathbb{E}_x G_1(x)^2 \leq \frac{\mu^2}{s}$. So, $G_1(x)^2 \leq \frac{\mu^2}{s}$ for some $x\in \mathbb{Z}_{p^m}^n$. For this particular $x$, $$P_Hf(x) = \widehat{f}(\chi_{r_i}) \chi_{r_i}(x) + G_1(x).$$ Therefore, 
$$
    |\widehat{f}(\chi_{r_i})| = |P_Hf(x) - G_1(x)|,
$$ 
which implies that $\widehat{f}(\chi_{r_i})$ is $\frac{\mu}{\sqrt{s}}$ close to $k$-granular, since $P_Hf(x)$ is $k$-granular and $|G_1(x)| \leq \frac{\mu}{\sqrt{s}}$.

\end{proof}

\subsection{Lower bound on the Fourier coefficients}\label{section_lower_bound}

In this section, we will prove \cref{granularity_of_Z_p} assuming \cref{lem:structure}, which gives us a lower bound on the Fourier coefficients of functions from $\G$ to $\{-1,+1\}$. This is a generalization on the granularity of a function $f:\mathbb{Z}_2^n \to \mathbb{R}$ when the domain of the function is $\G$. We will use the following theorem; for more details see~\cref{sec:nt}.

\begin{theorem}\label{theorem_roots_bound}
For $n \in \Z$, let $\omegan$ be a primitive root of unity. Let $f\in \mathbb{Z}[x]$, such that $f(\omega_n)\neq 0$. Then,
\[
    \bigg| \prod_{i \in \Z_n^{\ast}} f(\omegan^i) \bigg| \;\ge\;1\;. 
\]
\end{theorem}

\begin{proof}[Proof of \cref{granularity_of_Z_p}]
If we put $\mu=0$ in \cref{lem:structure}, we get that there exists a $g \in \Z[X]$, such that $|\widehat{f}(\chi)| \ge |g(\omega_{\lcmG})/\lcmG^{k}|$, where $k= \lceil \log_\lcmG 2s_f \rceil$, $s_f$ being the sparsity of the Boolean valued function $f$. We know by~Theorem~\ref{theorem_roots_bound} that  $|\prod_{i\in \mathbb{Z}_{\lcmG}^\times} g(\omega_{\lcmG}^i)| \ge 1$. Therefore, 
\begin{align*}
    \bigg|\prod_{i\in \mathbb{Z}_{\lcmG}^\times} g(\omega_{\lcmG}^i)\bigg| \;\ge\; 1 &\Rightarrow \prod_{i\in \mathbb{Z}_{\lcmG}^\times} \bigg|\frac{g(\omega_{\lcmG}^i)}{\lcmG^{k}}\bigg| \;\geq\; \left( \frac{1}{\lcmG^{k}} \right)^{\varphi(\lcmG)} \\
    &\Rightarrow \bigg|\prod_{i\in \mathbb{Z}_{\lcmG}^\times} \frac{g(\omega_{\lcmG}^i)}{\lcmG^{k}}\bigg| \;\geq\; \frac{1}{\lcmG^{k\varphi(\lcmG)}}.
\end{align*}


Now, the conjugate of $g(\omega_{\lcmG}^i)$, namely~$\overline{g(\omega_{\lcmG}^i)}$, is nothing but $=g(\omega_{\lcmG}^{\lcmG-i})$. Since $|g|_1 \le \lcmG^{k}$, it follows that for any $i \in \mathbb{Z}_{\lcmG}^\times$, $|g(\omega_{\lcmG}^i)/\lcmG^{k}| \le 1$. Therefore, 
    \begin{align}\label{eqn_2}
    \bigg|\prod_{i\in \mathbb{Z}_{\lcmG}^\times} g(\omega_{\lcmG}^i)\bigg| \;\ge\; 1 &\Rightarrow \prod_{i\in \mathbb{Z}_{\lcmG}^\times : i \leq \lfloor \frac{\lcmG}{2} \rfloor} |g(\omega_{\lcmG}^i)|^2 \;\geq\; 1 \nonumber \\
    &\Rightarrow \prod_{i\in \mathbb{Z}_{\lcmG}^\times : i \leq \lfloor \frac{\lcmG}{2} \rfloor} \bigg|\frac{g(\omega_{\lcmG}^i)}{\lcmG^{k}}\bigg|^2 \;\geq\; (\frac{1}{\lcmG^{k}})^{\varphi(\lcmG)} \nonumber \\
    &\Rightarrow \bigg|\frac{g(\omega_{\lcmG})}{\lcmG^{k}}\bigg|^2 \;\geq\; \frac{1}{\lcmG^{k\varphi(\lcmG)}} \nonumber \\
    &\Rightarrow \bigg|\frac{g(\omega_{\lcmG})}{\lcmG^{k}}\bigg| \;\geq\; \frac{1}{\lcmG^{k\varphi(\lcmG)/2}}.
    \end{align}

So, from \cref{eqn_2}, since $k = \lceil \log_{\lcmG} 2s_f \rceil$, we have $$\bigg|\frac{g(\omega_{\lcmG})}{\lcmG^{k}}\bigg| \;\geq\; \frac{1}{2^{\varphi(\lcmG)/2} s_f^{\varphi(\lcmG)/2}}.$$
\end{proof}

\begin{remark}\label{remark_Fourier_lower_bound}
     The proof-technique of \cref{granularity_of_Z_p} gives the following. If the Fourier coefficients of a Boolean function $f$ are $k$ granular, then  the Fourier coefficients of $f^2$ are $2k$-granular, and their absolute values are $\geq \frac{1}{2^{\varphi(\lcmG)} s_f^{\varphi(\lcmG)}}$.

\end{remark}

\subsection{Proof of \texorpdfstring{\cref{structure_theorem_2}}{}}\label{section_structure_theorem_2}

We have proved in \cref{section_structure_theorem_1} that if the Fourier weight of the function $f$ is concentrated on the $s$-largest Fourier coefficients, then $f$ is $\frac{\mu}{\sqrt{s}}$ close to a $k$-granular in the $\ell_2$ norm, where $k= \lceil \log_{\lcmG} 2s \rceil$. In this subsection, we will prove \cref{structure_theorem_2}, which says that this function must be a Boolean valued function.


\begin{proof}
    Let us assume $k= \lceil \log_{\lcmG} 2s \rceil$. Let 
    $B$ be the set of character functions corresponding to the $s$ largest coefficients and $S$ be set of the rest of the characters.
    Then, by \cref{lem:structure}, for all $\chi_{r_i} \in B$, $\widehat{f}(\chi_{r_i})$ is $\frac{\mu}{\sqrt{s}}-$close to $k$-granular. Therefore, for all $\chi_{r_i}\in B$, we can write $f(\chi_{r_i})$ as the sum of $\widehat{F}(\chi_{r_i})$ and $\widehat{G}(\chi_{r_i})$, that is, $$\widehat{f}(\chi_{r_i}) = \widehat{F}(\chi_{r_i}) + \widehat{G}(\chi_{r_i}),$$ where $|\widehat{F}(\chi_{r_i})|$ is $k$-granular and $|\widehat{G}(\chi_{r_i})| \leq \frac{\mu}{\sqrt{s}}$. 

    Now, for $\chi \in \mathcal{S}$, let us set $\widehat{F}(\chi)=0$ and $\widehat{G}(\chi) = \widehat{f}(\chi)$. Therefore, we can write $f(x)$ as sum of $F(x)$ and $G(x)$, that is, $$f(x)=F(x)+G(x).$$ Clearly, $F$ is $s-$sparse and $|\widehat{F}(\chi)|$ is $k$-granular for all Fourier coefficients $\widehat{F}(\chi)$ of $F$. 
    
    Let us look at $G(x)$. We have, 
    \begin{align*}
        \mathbb{E}[G(x)^2] 
        &\;=\; \mathbb{E}[\sum_{\chi\in \mathcal{B}} \widehat{G}(\chi)^2 + \sum_{\chi\in \mathcal{S}} \widehat{G}(\chi)^2] \\
        &\;=\; \mathbb{E}[\sum_{\chi_{r_i}\in \mathcal{B}} \widehat{G}(\chi_{r_i})^2 + \sum_{\chi\in \mathcal{S}} \widehat{f}(\chi)^2] \\
        &\;\leq\; s\cdot \frac{\mu^2}{s} + \mu^2 \\
        &\;\leq\; 2\mu^2.
    \end{align*}

\begin{claim}\label{claim:range-F}
Range of $F$ is \{-1,+1\}. 
\end{claim}



\begin{proof}[Proof of \cref{claim:range-F}]
        Since $f$ is Boolean valued and $f=F+G$, so we have, 
        \begin{align}\label{eqn_range_F}
            1=f^2=(F+G)^2 = F^2+2FG+G^2 = F^2 + G(2F+G) = F^2 + G(2f-G).
        \end{align}

        Let $H=G(2f-G)$. Let us look at the Fourier coefficients $\widehat{H}(\chi)$ of $H$.
        \begin{align}\label{eqn_H}
            |\widehat{H}(\chi)| &\leq ||G||_2||2f-G||_2 &\text{ by } \cref{lemma_bound_H} \nonumber \\
            &\leq ||G||_2 (||2f||_2 + ||G||_2) \nonumber \\
            &\leq \sqrt{2} \mu \cdot 2 \cdot 1 + 2\mu^2 \nonumber \\
            &<4\mu \nonumber \\
            &\leq \frac{1}{2 \times 2^{\varphi(\lcmG)} s^{\varphi(\lcmG)}}.
        \end{align}

        Now, since $\widehat{F}(\chi)$ is $k$-granular, any Fourier coefficient $\widehat{F^2}(\chi)$ of $F^2$ is $2k$-granular. That is, $\widehat{F^2}(\chi) = \frac{g'(\omega_{\lcmG})}{\lcmG^{2k}}$ and $g'(X) \in \mathbb{Z}[X]$ is a nonzero polynomial with degree $\leq \lcmG-1$. Also, by \cref{remark_Fourier_lower_bound}, $|\widehat{F^2}(\chi)| \geq \frac{1}{2^{\varphi(\lcmG)} s^{\varphi(\lcmG)}}$.

        Now, from \cref{eqn_range_F}, we have, $$\widehat{F^2}(\chi_0) + \widehat{H}(\chi_0)=1$$ and $$\widehat{F^2}(\chi) + \widehat{H}(\chi)=0, \text{ for all } \chi\neq \chi_0,$$
        where $\chi_0$ is the character which takes the value $1$ at all points in $\G$, as defined in \cref{lemma_character}, part 1.

        Since $\mid \widehat{F^2}(\chi) \mid \geq \frac{1}{2^{\varphi(\lcmG)} s^{\varphi(\lcmG)}}$ and $\mid \widehat{H}(\chi) \mid < \frac{1}{2\times 2^{\varphi(\lcmG)} s^{\varphi(\lcmG)}}$, these two cannot sum up to $0$; therefore, $\widehat{F^2}(\chi)=0$, $\widehat{H}(\chi)=0$ for all $\chi \neq 0$. 

        Now, let us look at $\widehat{F^2}(\chi_0)$. From the above, we have
        \begin{align*}
            &\, \widehat{F^2}(\chi_0) \,+\, \widehat{H}(\chi_0) \;=\;1 \\ \Rightarrow &\, \widehat{H}(\chi_0) \;=\; 1 \,-\, \frac{g'(\omega_{\lcmG})}{\lcmG^{2k}} \\
            \Rightarrow &\, \widehat{H}(\chi_0)\;=\; \frac{\lcmG^{2k} \,-\, g'(\omega_{\lcmG})}{\lcmG^{2k}}\;,
        \end{align*}
        which implies that $\widehat{H}(\chi_0)$ is $2k$-granular if $\widehat{H}(\chi_0) \neq 0$. So, by \cref{remark_Fourier_lower_bound}, $|\widehat{H}(\chi_0)| \geq \frac{1}{2^{\varphi(\lcmG)} s^{\varphi(\lcmG)}}$, which is not possible since $|\widehat{H}(\chi_0)| \leq \frac{1}{2 \times 2^{\varphi(\lcmG)} s^{\varphi(\lcmG)}}$ by \cref{eqn_H}. So, $\widehat{H}(\chi_0) = 0$, hence $\widehat{F^2}(\chi_0)=1$. Therefore $F^2(x)=0$ for all $x \in \G$, which implies that the range of $F$ is $\{-1,+1\}$.
    \end{proof}

    Now, let us find $\Pr_x[f(x)\neq F(x)]$. Since $F$ and $f$ both are Boolean valued, so the range of $G$ is equal to $\{-2,0,+2\}$. So, 
    \begin{align*}
        \Pr_x[f(x)\neq F(x)] &= \Pr_x[|G(x)|=2] \\
        &= \frac{1}{4} \mathbb{E}_x [G(x)^2] &\text{ by Markov's inequality} \\
        &\leq \frac{\mu^2}{2}.
    \end{align*}

    Therefore, 
    the $\ell_2$ distance between $f$ and $F$ is $\leq \sqrt{4\times \frac{\mu^2}{2}} = \sqrt{2}\mu$.

\end{proof}

\section{A \texorpdfstring{$\lcmG$}{}-independent lower bound}\label{section_p_independent_lower_bound}
In this section, we discuss an interesting property of Boolean-valued functions from $\G$ to $\{-1,+1\}$, 
and eventually establish a $\lcmG$-independent lower bound on the Fourier coefficient of an $s$-sparse Boolean-valued function. 

For $x \in \G$ and $i \in \mathbb{Z}_\lcmG^\times$, we define $i^{-1}x := (i^{-1}x_1,\hdots, i^{-1}x_n)$, where $x=(x_1,\hdots,x_n)$ and $i^{-1}$ denotes the inverse of $i$ in the multiplicative group $\mathbb{Z}_\lcmG^\times$ of $\mathbb{Z}_\lcmG$. 

\begin{lemma}\label{lemma_constructed_func_same_sparsity}
Let $f:\G \to \{-1,+1\}$, and $\chi\in \supp(f)$. Then, for each $i\in \mathbb{Z}_\lcmG^\times$, there exists a  function $h_i:\G \to \{-1,+1\}$, defined by $h_i(x)=f(i^{-1}x)$ such that 
\begin{enumerate}
    \item $\hat{h}_i(\chi) = \hat{f}(\chi^i)$,
    \item $\hat{h}_i(\chi) \ne 0 \iff \hat{f}(\chi) \ne 0$,
\end{enumerate}
where $i^{-1}$ denotes the inverse of $i$ in the multiplicative group $\mathbb{Z}_\lcmG^\times$ of $\mathbb{Z}_\lcmG$.
\end{lemma}

\begin{proof}
{\noindent \bf Proof of Part (1).} Let $h_i(x) := f(i^{-1}x)$. Note that, by definition, $h_i$ is a Boolean-valued function. Further, by Definition~\ref{def:foruier-transform}:
\begin{align*}
    \hat{h_i}(\chi)&=\frac{1}{|\G|}\sum_{x\in \G} h_i(x) \overline{\chi(x)} \\
    &=\frac{1}{|\G|}\sum_{ix\in \G} h_i(ix) \overline{\chi(ix)} \\
    &=\frac{1}{|\G|} \sum_{ix\in \G} f(x) \overline{\chi(x)^i} \\
    &=\frac{1}{|\G|} \sum_{x\in \G} f(x) \overline{\chi(x)}^i \\
    &=\hat{f}(\chi^i).
\end{align*}

{\noindent \bf Proof of Part (2).} We want to prove the following.
\begin{claim}\label{claim_equal_sparsity}
$\hat{h_i}(\chi) \neq 0 \text{ iff } \hat{f}(\chi)\neq 0$.
\end{claim}
\begin{proof}
By Definition~\ref{def:foruier-transform}, we have $\hat{f}(\chi)=\frac{1}{|\G|} \sum_{x\in \G} f(x) \overline{\chi(x)} =\frac{u(\omega|\lcmG)}{|\G|}$, where $u(\omega_\lcmG)=\sum_{x\in \G} f(x) \overline{\chi(x)}$. Think of $u(X) \in \Z[X]$. Then, $\hat{h_i}(\chi)=\hat{f}(\chi^i)=\frac{1}{|\G|} \sum_{x\in \G} f(x) \overline{\chi(x)}^i = \frac{u(\omega_\lcmG^i)}{|\G|}$; the last holds since 
\[
u(\omega_\lcmG)\;=\;\sum_{x\in \G} f(x) \overline{\chi(x)} \;\iff\; u(\omega_\lcmG^i)\;=\;\sum_{x\in \G} f(x) \overline{\chi(x)}^i\;.
\]
Using Fact~\ref{fact:non-zero-anypower}, we know that $u(\omega_\lcmG) \ne 0 \iff u(\omega_\lcmG^i)\neq 0$. Hence, whenever $f(\chi)\neq 0, \ h_i(\chi)\neq 0$ as well, and the vice versa. Thus, the claim holds.
\end{proof}
\end{proof}



The following corollary follows immediately from \cref{lemma_constructed_func_same_sparsity} (part 2). 

\begin{corollary}\label{cor_equal_sparsity}
    $s_f=s_{h_i}$, where $s_{h_i}$ is the sparsity of the function $h_i$.
\end{corollary}

\begin{proof}
From the above claim, one concludes that the function $h_i$ has the same sparsity as the function $f$, that is, $s_{h_i}=s_f$.






\end{proof}

The following corollary follows immediately from \cref{lemma_constructed_func_same_sparsity}.

\begin{corollary}\label{cor_bound_ind_sparsity}
For all function $f:\G\to \{-1,+1\}$, $s_f \ge p-1$.
\end{corollary}

\begin{proof}
From the proof of \cref{lemma_constructed_func_same_sparsity}, it is clear that if $\chi\in \supp(f)$, then $\chi^i\in \supp(f)$ for all $i\in \mathbb{Z}_\lcmG^\times$. Hence, sparsity of $f, \ s_f \geq |\mathbb{Z}_\lcmG^\times| = \varphi(\lcmG)$.
\end{proof}

Now, we prove a lower bound on the Fourier coefficient which is $\lcmG$-independent; this is an immediate consequence of \cref{cor_bound_ind_sparsity}.

\begin{theorem}[$\lcmG$-independent lower bound] \label{thm:p-independent-lb}
    Let $f:\G \to \{-1,+1\}$ be a function. Let $\chi \in \supp(f)$. Then,
    \[
    |\widehat{f}(\chi)| \;>\; \frac{1}{((s_f+1)\sqrt{s_f})^{s_f}}\;.\]  
\end{theorem}

\begin{proof}
We know that $s_f \geq \varphi(\lcmG)$, by \cref{cor_bound_ind_sparsity}. Hence, we combine this with Theorem~\ref{granularity_of_Z_p} , to get
\[
|\widehat{f}(\chi)| \;\geq\; \frac{1}{2^{\varphi(\lcmG)/2} s^{\varphi(\lcmG)/2}}\;>\;\frac{1}{2^{s_f/2} s^{s_f/2}}.\]
\end{proof}

\section{Sparse Boolean-valued function with small Fourier coefficients}\label{sec:examples}
In this section, for a fixed prime $p \geq 5$, and arbitrarily large $s$, we give an example of a function $f:\mathbb{Z}_p^n \to \{-1,+1\}$, such that the minimum of the absolute value of its Fourier coefficients is at most $o\left( 1/s \right)$. More specifically, we give the details of the construction of \cref{thm:small-fourier}.

To prove  Theorem~\ref{thm:small-fourier} we define a function, which is basically composition of $\AND_n$ and univariate Threshold functions; we call $\AT$.

\begin{definition}
Let us define function $\AT:\mathbb{Z}_p^n \to \{-1,+1\}$ by 
\[
\AT(x_1,x_2,\ldots,x_n)\;:=\;\AND_n \left(\mathbb{I}_{\geq \frac{p+1}{2}}(x_1), \mathbb{I}_{\geq \frac{p+1}{2}}(x_2), \ldots ,\mathbb{I}_{\geq \frac{p+1}{2}}(x_n)\right),\]
where the univariate Threshold function $\mathbb{I}_{\geq \frac{p+1}{2}} : \Z_p \to \{-1, +1\}$, is defined as:
\[
\mathbb{I}_{\geq \frac{p+1}{2}}(x)\;=\;
\begin{cases}
1 & \text{ for } x\geq \frac{p+1}{2} \\
-1 & \text{ otherwise.}
\end{cases}
\]

\end{definition}

\begin{lemma}\label{lemma_counterexample}
There is a Fourier coefficient of $\AT$, whose absolute value is $\frac{1}{p^{nc}}$, where $c$ is a constant $>1$.
\end{lemma}

\begin{proof}
Let us look at the coefficient $\widehat{I_{\geq \frac{p+1}{2}}}(\chi_2)$, where $\chi_2(x)=\omegap^{2\cdot x}$. By definition,

\begin{align*}
     \widehat{I_{\geq \frac{p+1}{2}}}(\chi_2)\;&=\;\frac{1}{p} (-1+\omega_p-\omega_p^2+\omega_p^3-\cdots) \\ &=\; -\frac{1}{p} \cdot \frac{1+\omegap^p}{1+\omegap} = -\frac{2}{p}\cdot \frac{1}{1+\omegap}\;.
\end{align*}
where we simply used the GP series with the ratio $-\omegap$. Note that 
\[
|1+\omegap|^2 = (1+\cos(2\pi/p))^2 + \sin^2(2\pi/p) = 2 (1+\cos(2\pi/p))= 4 \cos^2 (\pi/p)\;.
\]
Therefore, 
\[
\left|\widehat{I_{\geq \frac{p+1}{2}}}(\chi_2)\right| = \frac{1}{p \cos(\pi/p)}\;. 
\]
We are not using mod, because for $p \ge 5$, $\cos(\pi/p) > 0$. Now, let us consider the function $\AND_n: \mathbb{Z}_2^n \to \{-1,+1\}$ defined by 
$$\AND_n (x_1,x_2,\ldots,x_n)= 
\begin{cases}
   -1 & \text{ if } x_i=1 \text{ for all } i \\
   1 & \text{otherwise}.
\end{cases}$$

It is already known that the coefficient of $(-1)^{x_1+ \cdots +x_n}$ is $\frac{(-1)^{n-1}}{2^{n-1}}$. Now, let us look at the function 
$$
    \AT(x_1,x_2,\ldots,x_n)=\AND_n \left(\mathbb{I}_{\geq \frac{p+1}{2}}(x_1), \mathbb{I}_{\geq \frac{p+1}{2}}(x_2), \ldots ,\mathbb{I}_{\geq \frac{p+1}{2}}(x_n)\right).
$$
We know that other than the constant coefficient, the absolute value of all other coefficients of $\AND_n$ is $\frac{1}{2^{n-1}}$. So, in particular,
\[
    |\widehat{\AT}(\chi_{2,2,\cdots2})| \;=\; \bigg ( \frac{1}{p \cos(\pi/p)} \bigg )^n \cdot \frac{1}{2^{n-1}} \;=\; 2 \bigg ( \frac{1}{2p \cos(\pi/p)} \bigg )^n\;.
\]
Comparing $p^{nc}$ with the above, one gets that 
\[
 \frac{1}{p^{nc}}\;=\; 2\bigg ( \frac{1}{2p \cos(\pi/p)}\bigg )^n 
\; \implies\; c = \log_p |2 \cos (\pi/p)| +1-\frac{1}{n}\log_{p} 2.
\]   
Note that, for large enough $n$, and a fixed $p \ge 5$, $c > 1$. Hence, the proof follows.
\end{proof}

Now we are ready to prove Theorem~\ref{thm:small-fourier}.
\begin{proof}[Proof of Theorem~\ref{thm:small-fourier}]
This directly follows, as we can claim from \cref{lemma_counterexample} that there exists a family of functions in $\mathbb{Z}_p^n$ whose absolute value of the minimum coefficient is not linear in $\frac{1}{\text{sparsity}}$, but actually $=\Omega(\frac{1}{\text{sparsity}^{1+\epsilon_p}})$. where $\epsilon_p > 0$, is a $p$-dependent constant.
\end{proof}


\begin{lemma}\label{lemma_counterexample_2}
    There exists a function whose one of the Fourier coefficients is less than that of $\AT$.
\end{lemma}

\begin{proof}
    Consider the function $f:\mathbb{Z}_5^2 \to \{-1,+1\}$ whose truth table is given in Table~\ref{wrap-tab:1}.

Observe that:
\begin{itemize}
    \item $\widehat{f}(\chi_{1,0}) = \frac{1}{25} (-5+5\omega_5-5\omega_5^2+\omega_5^3 +\omega_5^4) = \frac{0.29}{25}.$
    \item $s_f=25$, where $s_f$ is the sparsity of $f$.
\end{itemize}

Therefore, $$\widehat{f}(\chi_{1,0}) = \frac{1}{s_f^{1+0.385}}.$$ So comparing with \cref{c_value_Z_5}, we can see that $c$ is bigger in this case.
\end{proof}

\begin{remark}
When $p=3$, we know that $2 \cos (\pi/3)=1$, and hence $c < 1$ ($c \to 1$, when $n \to \infty$). 
\end{remark}

\begin{remark}
    When $p=5$, and $n=2$, we have, 
    \begin{align}\label{c_value_Z_5}
        c= \log_5 |(2\cos(\frac{\pi}{5}))|+1 - \frac{1}{2} \cdot \log_5 (2)=1.265...
    \end{align}
\end{remark}

\begin{wraptable}{r}{8.0 cm}
\caption{Truthtable for $f : \Z^{2}_{5} \rightarrow \{-1, +1\}$.}\label{wrap-tab:1}
\vspace{-20pt}
\begin{center}
\begin{tabular}{lll}\\ \toprule
$x_1$ & $x_2$ & $f(x_1,x_2)$ \\ \toprule
 0 & 0 & \quad -1  \\
 0 & 1 & \quad ~1 \\
 0  &  2  &  \quad -1 \\
 0  &  3  &  \quad ~1 \\
 0  & 4 & \quad -1 \\
 1  & 0 & \quad -1 \\
 1  & 1 & \quad ~1 \\
 1  & 2 & \quad -1 \\
 1  & 3 & \quad ~1 \\
 1  & 4 & \quad -1 \\
 2  & 0 & \quad -1 \\
 2  & 1 & \quad ~1 \\
 2  & 2 & \quad -1 \\
 2  & 3 & \quad ~1 \\
 2  & 4 & \quad -1 \\
 3  & 0 & \quad -1 \\
 3  & 1 & \quad ~1 \\
 3  & 2 & \quad -1 \\
 3  & 3 & \quad -1 \\
 3  & 4 & \quad -1 \\
 4  & 0 & \quad -1 \\
 4  & 1 & \quad ~1 \\
 4  & 2 & \quad -1 \\
 4  & 3 & \quad -1 \\
4 & 4 & \quad -1 \\
\bottomrule
\end{tabular}
\end{center}
\vspace{-10pt}
\end{wraptable}

\section{Testing \texorpdfstring{$s$}{s}-sparsity}\label{section_thm:algo-sparse-check}

In this section our goal is to design an algorithm (see next page) that given a function $f:\G \to \{-1,+1\}$ accepts (with probability $\geq \frac{2}{3}$) if $f$ is $s$-sparse and rejects (with probability $\geq \frac{2}{3}$) if $f$ is $\epsilon$-far from any $s$-sparse Boolean function.

The algorithm~\ref{algo:sparsity}'s primary concept closely aligns with that presented in \cite{gopalan2011testing}. Since we know from \cref{granularity_of_Z_p} that all the coefficients of an $s$-sparse  function are more than $\frac{1}{2^{\varphi(\lcmG)/2} s_f^{\varphi(\lcmG)/2}}$, so the idea is to {\em partition the Dual group} (the set of characters) into $\Theta(s^2)$ buckets, and for each bucket $B$ estimate the weight $\wt(B)$ of the bucket which is defined as follows: 

\begin{definition}[The weight function]
    Let $f:\G \to \{-1,+1\}$ be a function. The weight of a bucket $B$ is defined by $$wt(B)\;:=\; \sum_{\chi \in B} |\widehat{f}(\chi)|^2.$$
\end{definition}

 We will try to estimate the weight of each bucket within an additive error of $\tau/3$, where $\tau \geq \frac{1}{2^{\varphi(\lcmG)} s_f^{\varphi(\lcmG)}}$. If the function $f$ is indeed $s$-sparse, then clearly the number of buckets with estimated weight more than $2\tau/3$ is at most $s$. On the other hand, we will show that if the function $f$ is $\epsilon$-far from any $s$-sparse Boolean function, then with high probability at least $(s+1)$ buckets will have estimated weight more than $2\tau/3$. 

Thus, the main problem boils down to estimating the weight of a bucket by querying $f$ at a {\em small} number of points. Unfortunately, for arbitrary buckets the number of queries to be made can be a lot! To bypass this, the buckets are choosen carefully. The buckets corresponds to the cosets of $H^{\perp}$ where $H$ is a random subspace of $\G$ of dimenstion, $t = \Theta(s^2)$.
For such kind of buckets we show that estimation of the weight can be done using a small number of samples.

\newpage
\begin{algorithm}\caption{Test Sparsity}\label{algo:sparsity}

\begin{algorithmic}[1]
\State \textbf{Input:} $s, \epsilon$, and query access to $f: \G \to \{-1,+1\}$. 
\State \textbf{Output}: YES, if $f$ is $s$-sparse, and NO, if it is $\epsilon$-far from any Boolean valued function. 

\bigskip
\State \textbf{Parameter setting}: Set $t:= \lceil 2 \log_{\lcmG} s + \log_\lcmG 20 \rceil +1,\tau:=\min(\frac{\epsilon^2}{40 \lcmG^t}, \frac{1}{(\lcmG^2s)^{\varphi(\lcmG)}})$ 
$M := O(\log(\lcmG^t) \cdot \frac{1}{\tau^2})$ 

\bigskip

\State Choose $v_1, \dots, v_t$ independent elements uniformly at random from $\G$. \label{ln:H1}

\State Let $H=\Span\{v_1, \dots, v_t\}$ \label{ln:H2}

\State Pick $(z_1, x_1), \dots, (z_M, x_M)$ uniformly and independently from $H\times \G$ \label{ln:query1}

\State Query $f(x_1), \dots, f(x_M)$ and $f(x_1 - z_1), \dots, f(x_M -z_M)$ \label{ln:query2}

\For{For every $r \in \G$} \label{ln:forstart}

\State Let $\frac{1}{M}\sum_{i = 1}^M \chi_r(z_i)f(x_i)f(x_i - z_i)$ be the estimate of $wt(r + H^{\perp})$  \label{ln:estimation}

\EndFor

\If{number of $r$ for which the estimate of  $wt(r + H^{\perp})$ is $\geq \frac{2\tau}{3}$ is $\leq s$} 
\State Output YES \label{ln:output1}
\Else 
\State Output NO \label{ln:output2}
\EndIf

\end{algorithmic}
\end{algorithm}

\subsection{The Algorithm}\label{sec:algorithm}



\cref{algo:sparsity} takes a Boolean valued function $f:\G \to \{-1,+1\}$ as input. It accepts with probability $\geq \frac{2}{3}$ if $f$ is $s$-sparse, and rejects with probability $\geq \frac{2}{3}$ if $f$ is $\epsilon$-far from any $s$-sparse Boolean valued function.

In lines \ref{ln:H1} and \ref{ln:H2} we pick a random vector space of $\G$ of dimension $t$, spanned by $v_1, \hdots, v_t$.

For each $b=(b_1,\ldots,b_t)$ in $\mathbb{Z}_\lcmG^t$, where $\lcmG = \LCM\{p_1^{m_1}, \ldots, p_T^{m_T}\}$, we define the bucket $C(b)$ as follows:
    \begin{align*}
        C(b):=\{ \chi\in \G : \langle \chi,v_i \rangle = b_i \ \forall i\in \{1,2,\ldots, t\}\}.
    \end{align*}
We perform a random shift $u\in \mathbb{Z}_\lcmG^t$ (~\ref{def:random-shift}) and rename the buckets $C(b)$ as $C_u(b)$. Observe that the buckets $C_u(b)$ gives a partition of $\G$ and each $C_u(b)$ corresponds to a coset in $\G/H^{\perp}$. Thus once we have picked the random vector space $H$ in lines \ref{ln:H1} and \ref{ln:H2} we have a partition of $\G$ (and thus the Dual space of $\G$) into $\lcmG^t$ buckets where each bucket is of the form $(r + H^{\perp})$ for $r\in \G$.

Now in lines \ref{ln:query1} to \ref{ln:estimation} we estimate the weight of each buckets. The following lemma proved the correctness of the estimation.
\begin{lemma}[Estimating the wt function] \label{lem:estimate-wt}
    For all $r \in \G$,
    the value $\frac{1}{M}\sum_{i = 1}^M \chi_r(z_i)f(x_i)f(x_i - z_i)$ calculated in Line~\ref{ln:estimation} satisfy the following:
    $$
        \Pr \left(\left|\wt(r + H^{\perp}) - \frac{1}{M}\sum_{i = 1}^M \chi_r(z_i)f(x_i)f(x_i - z_i) \right|\leq \frac{\tau}{3} \right) \geq 1-\frac{1}{10\lcmG^t}.
    $$
\end{lemma}
%


The proof of \cref{lem:estimate-wt} is given in \cref{sec:estimation}. From \cref{lem:estimate-wt} we observe that by union bound (since there are $\lcmG^t$ buckets) at the end of line \ref{ln:estimation} with probability at least $\frac{9}{10}$ for all $r\in \G$ the value of $\wt(r + H^{\perp})$ is estimated within additive $\frac{\tau}{3}$. That is, with probability $\frac{9}{10}$ 
\begin{equation}\label{eq:estimation}
    \forall r \in \G,\ \left|\wt(r + H^{\perp}) - \frac{1}{M}\sum_{i = 1}^M \chi_r(z_i)f(x_i)f(x_i - z_i) \right|\leq \frac{\tau}{3}\;.
\end{equation} 

Now we can prove the completeness and soundness of our algorithm. 

\begin{lemma}[Completeness theorem]\label{lem:sparse-acceptance}
    If $f:\G\to \{-1, +1\}$ is an $s-$sparse Boolean function, then with probability $\frac{2}{3}$, the test accepts.
\end{lemma}

\begin{proof}
    Let $f= \sum_{\alpha} \hat{f}(\alpha)\chi_{\alpha}$. Since, $f$ is at most $s$-sparse, the support set $A:=\{\alpha \mid \hat{f}(\alpha) \ne 0\}$  has cardinality at most $s$. So clearly, there are at most $s$ buckets (that is at most $s$ different  $r$) such that $\wt(r + H^{\perp})$ is non-zero. Also by \cref{granularity_of_Z_p} we know that each non-zero bucket must have weight at least $\tau$. 
    
    It thus follows from \cref{eq:estimation} that with probability $\frac{9}{10}$,  for each bucket with zero weight the estimation is $\leq \frac{\tau}{3}$ and for bucket with non-zero weights the estimation is $\geq \frac{2\tau}{3}$. Thus, in line \ref{ln:output1} the algorithm will output YES with probability at least $\frac{9}{10}$. 
\end{proof}

We now prove the soundness of our algorithm. 
We wish to show that if $f$ is $\epsilon$-far from any $s$-sparse Boolean function, then the test outputs NO with high probability. For this, we need to prove that in this case there are at least $s+1$ buckets with weight of more than $\tau$. This is proved in the following lemma. 

\begin{lemma}\label{lem:bucket}
    If $f:\G\to \{+1, -1\}$ is $\epsilon$-far from any $s$-sparse Boolean function, where $\epsilon \leq \frac{1}{8(2s)^{\varphi(\lcmG)}}$, then with probability $\frac{9}{10}$ there are at least  $s+1$ buckets (that is at least $s+1$ vectors $r$ in $\G$) such that weights of these buckets are more that $\tau$. 
\end{lemma}

The proof of \cref{lem:bucket} is given in \cref{sec:bucket}. Using \cref{lem:bucket} we have the soundness of our algorithm.

\begin{lemma} \label{lem:rejection-prob}
    If $f$ is $\epsilon$-far from any $s$-sparse Boolean function, where $\epsilon \leq \frac{1}{8(2s)^{\varphi(\lcmG)}}$, the test rejects with probability $\frac{2}{3}$. 
\end{lemma}

\begin{proof}
    Let $f$ be $\epsilon$-far from any $s$-sparse Boolean function. By \cref{lem:bucket} with probability $\frac{9}{10}$ there are at least $(s+1)$ buckets with weight at least $\tau$. 

    From \cref{eq:estimation}, with probability $\frac{9}{10}$ the estimations of all the bucket are indeed accurate and thus for each bucket with weight $\geq \tau$ the estimation is $\geq \frac{2\tau}{3}$.  And so in line \ref{ln:output2} the algorithm will output NO.   
\end{proof}

Finally we observe that the algorithm makes $O(M) = O(\log(\lcmG^t) \cdot \frac{1}{\tau^2}) = \poly((2s)^{\varphi(\lcmG)}, \frac{1}{\epsilon})$ samples.



\subsection{Proof of \texorpdfstring{\cref{lem:estimate-wt}}{}}\label{sec:estimation}


We start by proving the following lemma that proves that on expectation the output of line~\ref{ln:estimation} is $\wt(r + H^{\perp})$.

\begin{lemma}\label{lem:expectation}
    $\E_{x\in \G, z \in H} [\chi_r(z)f(x)f(x-z)] = \wt(r + H^{\perp})$.
\end{lemma}

\begin{proof}[Proof of \cref{lem:expectation}]
For the proof of this lemma we need the projection operator defined in \cref{def:proj-op}. Recall given a subgroup $H^{\perp}\subseteq \G$ and a coset $r+H^{\perp}$, we define the projection operator $P_{r+H^{\perp}}f$ on functions $f:\G \to \mathbb{R}$ as follows:
\begin{align*}
        \widehat{P_{r+H^{\perp}}f} (\alpha) = 
        \begin{cases}
            \widehat{f}(\alpha)  & \alpha \in r+H^{\perp} \\
            0  & \text{otherwise}.
        \end{cases}
    \end{align*}

Recall that $wt(r + H^{\perp}) = \sum_{\alpha \in r+H^{\perp}} |\hat{f}(\alpha)|^2$ which is equal to $\sum_{\alpha \in \G} |\widehat{P_{r+H^{\perp}}f}(\alpha)|^2$, by the definition of $\widehat{P_{r + H^{\perp}}f}$. By Parseval, we know that $$\sum_{\alpha \in \G} |\widehat{P_{r+H^{\perp}}f}(\alpha)|^2 = \E_{w \in \G} [ P_{r+H}f(w) \overline{P_{r+H}f(w)}].$$ 

From \cref{lemma_granularity_Z_p} we have 
$P_{r+H^{\perp}} f(x) = \E_{z\in H} [f(x-z)\chi_r(z)]$ and hence $\overline{P_{r+H} f(x)} = \E_{z\in H^\perp} [f(x-z)\overline{\chi_r(z)}]$.  Thus,
    \[
\E_{w \in \G} [ P_{r+H}f(w) \overline{P_{r+H}f(w)}]\;=\; \E_{w \in \G, z_1,z_2 \in H^\perp} [ \chi_r(z_1)f(w-z_1) \overline{\chi_r(z_2)} f(w-z_2)]\;.
    \]
   Since, the conjugate of $w_\lcmG$ ($\lcmG$-th root of unity) is $w_\lcmG^{\lcmG-1}$, by definition of $\chi_r(\cdot)$, we have 
\[
\chi_r(z_1) \cdot \overline{\chi_r(z_2)} = \chi_r(z_1) \cdot \chi_r(-z_2) = \chi_r(z_1-z_2)\;.
\]
Define, $z:=z_1-z_2$, and $x:=w-z_2$. Then, $x-z=w-z_1$. Rewriting the above, we get 
\[
wt(r+H) = \E_{x\in \G, z \in H^\perp} [\chi_r(z)f(x)f(x-z)]\;
\]
as desired.
\end{proof}

We will be using the Hoeffding's Inequality (\cref{lem:chernoff-sample}) for our proof of correctness.

Now we are ready to prove \cref{lem:estimate-wt}.

\begin{proof}[Proof of \cref{lem:estimate-wt}]
Define $\delta := \frac{1}{10\lcmG^t}$.
Define $Y_{i} := \chi_r(z_i)f(x_i)f(x_i-z_i)$. 
Write $\chi_r(z_i) = a_r(z_i) + \iota b_r(z_i)$, where $a_r(z_i)$ and $b_r(z_i)$ are real numbers. Furthermore, since $|\chi_r(z_i)|=1$ (\cref{lemma_character} (4)), we have that $a_r(z_i), b_r(z_i) \in [-1, 1]$. Define, $Y_{1,i} := a_r(z_i)f(x_i)f(x_i-z_i)$, and $Y_{2,i} := b_r(z_i)f(x_i)f(x_i-z_i)$. Since $f$ is Boolean valued, we have that $Y_{1,i}, Y_{2,i} \in [-1, 1]$. By definition, $Y_{i} = Y_{1,i} + \iota Y_{2,i}$. 
Consider $Y_1:= \frac{\sum_{i \in [M]} Y_{1,i}}{M}$  and $Y_2:= \frac{\sum_{i \in [M]} Y_{2,i}}{M}$. Since the pairs $(x_i, z_i)$ are all drawn independently, we have that the random variables $\{Y_{1,i}\}_{i=1}^k$ are independent, and the random variables $\{Y_{2,i}\}_{i=1}^k$ are independent. By Hoeffding's Inequality (\cref{lem:chernoff-sample}) we have that for each $i \in [2]$,
\begin{align*}
    \Pr\left[\Big|Y_i - \E[Y_i]\Big| \ge \frac{\tau}{\sqrt{2}}\right]
\le\; &\Pr\left[\Big|\sum_{j \in [M]} Y_{i,j} - \E\Big[\sum_{j \in [M]} Y_{i,j}\Big]\Big| \geq \frac{\tau M}{\sqrt{2}}\right] \\ &\leq 2\exp\left(-\frac{\tau^2 M}{4}\right)\\
&\leq \frac{\delta}{2}\;,
\end{align*}
where the last inequality follows by an appropriate choice of the constant hidden in the definition $M=O(\log(\lcmG^t)\cdot\frac{1}{\tau^2})$ in \cref{algo:sparsity}.
Let $Y=Y_1 + \iota Y_2 = \frac{\sum_{i \in [n]} Y_{i}}{M}$. Note that, $(Y-\E[Y])^2= (Y_1-\E[Y_1])^2 + (Y_2 - \E[Y_2])^2$. Trivially, $|Y-\E[Y]| \ge \tau \implies \exists i \in [2]$, such that $|Y_i-\E[Y_i]| \ge \frac{\tau}{\sqrt{2}}$. Therefore, we have 
\[
\Pr[|Y - \E[Y]| \ge \tau] \le \Pr[|Y_1 - \E[Y_1]| \ge \frac{\tau}{\sqrt{2}}] + \Pr[|Y_2 - \E[Y_2]| \ge \frac{\tau}{\sqrt{2}}]\;\le\;\delta\;.\] 
From \cref{lem:expectation} we have that $\E[Y]= \wt(r + H^{\perp})$. Hence we have our lemma.
\end{proof}


\subsection{Proof of \cref{lem:bucket}}\label{sec:bucket}

Before we go on to prove \cref{lem:bucket} it is important to understand some structure of the partition of the dual set we obtain in Lines~\ref{ln:H2} and \ref{ln:estimation}. Recall that after we select the subspace $H$ at random the dual set is partitioned into buckets $C(b)$, where,  

\begin{align*}
    C(b)=\{ \chi\in \G : \langle \chi,v_i \rangle = b_i \ \forall i\in \{1,2,\ldots, t\}\},
\end{align*}
where $v_i$ are the rows of the matrix $A$.

Consider the following random variables: 
\begin{definition}
    The indicator random variable $I_{\chi \to b}$ for the event $\chi \in C(b)$ is defined by 
    \begin{align*}
        I_{\chi \to b} = 
        \begin{cases}
        1 & \mbox{ if } \chi \in C(b) \\
        0 & \text{otherwise}.
        \end{cases}
    \end{align*}
\end{definition}

\begin{definition}[Random shift] \label{def:random-shift}
    Let us pick a non-zero element $u \in \mathbb{Z}_\lcmG^t$. For each $b$, let us define $C(b+u)$ as $C_u(b)$. 
    This is known as a random shift. 
    It is easy to see that $C_u(b)$ is also a coset and is isomorphic to $C(b)$. 
\end{definition}

Recall the definition of covariance of two random variables (\cref{defn_covariance}). We need the following theorem.
\begin{theorem}\label{thm:random_shift-coset}
    Let us consider the partition in buckets $C(b)$. Then the following are true:

    \begin{enumerate}
        \item For all $\chi \in \G$ and for all $b \in \mathbb{Z}_\lcmG^t$, $\Pr_{u,A}[ \chi \in C_u(b)]= \frac{1}{\lcmG^{t}}$.
 
        \item Let $S\subseteq \G$ be such that $|S| \leq s+1$. Then, if $t \geq 2\log_\lcmG s + \log_\lcmG (\frac{1}{\delta})$, all $\chi \in S$ belong to different buckets except with probability at most $\delta$. 
        
        \item For all $b\in \mathbb{Z}_\lcmG^t$ and for all distinct $\chi, \chi' \in \widehat{\G}$, $\cov(I_{\chi \to b+u}, I_{\chi' \to b+u}) = 0$.
        
    \end{enumerate}
\end{theorem}

\begin{proof}
    \begin{enumerate}
    \item The event $\chi \in C_u(b)$ is equivalent to the event $\forall i \in [t], \langle \chi, v_i \rangle = b_i + u_i$. For each $i \in [t]$, $\Pr[\langle \chi, v_i \rangle = b_i + u_i]=\Pr[u_i=\langle \chi, v_i \rangle - b_i ]=1/\lcmG$. Since, $u_1, u_2,\ldots, u_t$ and $b_1, \ldots, b_t$ are all independent, the statement follows.
    
    \item Let $\chi$ and $\chi'$ be two distinct characters in $\widehat{\G}$. Also note that the number of buckets is $\lcmG^t$. We have that 
    \begin{align*}
        \Pr \left[ \chi, \chi' \text{ belong to the same bucket} \right]
        =\Pr \left[ \forall i \in [t], \langle\chi-\chi', v_i\rangle=0 \right]
        =\frac{1}{\lcmG^t},
    \end{align*}
    where the last equality holds because $\chi$ and $\chi'$ are distinct, and $v_1, \ldots, v_t$ are independent.
    
     Since $|S|\leq s+1$, therefore the number of ways in which two distinct $\chi, \chi'$ can be chosen from $S$ is $\binom{s+1}{2} \leq s^2$. Therefore, probability that all $\chi\in S$ belong to different buckets is given by
        \begin{align*}
            \Pr[\text{all } \chi\in S \text{ belong to different buckets}] &\geq 1- \binom{s+1}{2} \frac{1}{\lcmG^t} \\
            &\geq 1- s^2 \frac{1}{\lcmG^t} \\
            &\geq 1- s^2 \frac{1}{\lcmG^{2\log_\lcmG s + \log_\lcmG (\frac{1}{\delta})}} \\
            &= 1- s^2 \frac{1}{\frac{s^2}{\delta}} \\
            &= 1- \delta.
        \end{align*}

    \item 
        For all $b$, $\Pr_u[ b+u \neq 0] = \frac{\lcmG^t-1}{\lcmG^t}$. Now, let us look at the following cases to prove the theorem.
    \begin{description}
            
        \item [Case 1: $\chi'$ is a multiple of $\chi$, and both of them are nonzero.] 
       In this case, there exists a $\lambda \in \mathbb{Z}_\lcmG$ such that $\chi' = \lambda \chi$. Then,
       \begin{align*}
           &\langle \chi', v_i \rangle = \lambda \langle \chi, v_i \rangle \ \forall i \in [t] \\
           &\Rightarrow (b_i+u_i) (\lambda-1)=0 \ \forall i \in [t] \\
           &\Rightarrow (b+u) (\lambda-1)=0.
       \end{align*}

       We have the following to subcases.

       \begin{description}
           \item[Subcase 1: ($\lambda-1$) is invertible in $\mathbb{Z}_\lcmG$.] In this case, $I_{\chi \to b+u}\cdot I_{\chi' \to b+u}=1$ if and only if 
       \begin{enumerate}
       \item $b+u=0^t$, i.e., $u=-b$, and
       \item $\forall i \in [t], \langle \chi, v_i \rangle=b_i+u_i$.
       \end{enumerate}
       The probability of event (a) is $\frac{1}{\lcmG^t}$. Now, conditioned on event (a). Since the $v_1,\ldots, v_t$ and $u$ are all independent random variables, the conditional probability of event (b) is $\frac{1}{\lcmG^t}$. We thus have that $\E[I_{\chi \to b+u}\cdot I_{\chi' \to b+u}]=\frac{1}{\lcmG^t}\cdot \frac{1}{\lcmG^t}=\frac{1}{\lcmG^{2t}}$.

       Now, $\E[I_{\chi \to b+u}]=\Pr[\chi \in C_u(b)]=\frac{1}{\lcmG^t}$ by part (1). Similarly, $\E[I_{\chi' \to b+u}]=\frac{1}{\lcmG^t}$.
       We thus have that 
        $\cov(I_{\chi \to b+u}, I_{\chi' \to b+u}) = \frac{1}{\lcmG^{2t}}-\frac{1}{\lcmG^t}\cdot \frac{1}{\lcmG^t}=0$.

           \item[Subcase 2: ($\lambda-1$) is not invertible in $\mathbb{Z}_\lcmG$.]
           Then $\lambda-1$ is a zero divisor in $\mathbb{Z}_\lcmG$. It follows that $b_i+u_i= m_i$, for some $m_i$ such that $m_i$ divides $\lcmG$, for all $i \in [t]$.
        Observe that $I_{\chi \to b+u}\cdot I_{\chi' \to b+u}=1$ if and only if 
       \begin{enumerate}
       \item $b_i+u_i= m_i$ for all $i \in [t]$, i.e., $u_i=m_i-b_i$ for all $i \in [t]$, and
       \item $\forall i \in [t], \langle \gamma_1, \beta_i \rangle=b_i+u_i$.
       \end{enumerate}
       Clearly, $\Pr[u_i=m_i-b_i \ \forall i \in [t]] = \frac{1}{\lcmG^{t}}$, since the value of $\lambda$ determines that of $m_i$'s. Since the $\beta_1,\ldots, \beta_t$ and $u$ are all independent random variables, $\Pr[\forall i \in [t], \langle \chi, \beta_i \rangle=b_i+u_i | u_i=m_i-b_i] = \frac{1}{\lcmG^{t}}$. Therefore, $\mathbb{E}[I_{\chi \to b+u}\cdot I_{\chi' \to b+u}]=\frac{1}{\lcmG^{t}}\cdot \frac{1}{\lcmG^{t}}=\frac{1}{\lcmG^{2t}}$.

       Now, $\mathbb{E}[I_{\chi \to b+u}]=\Pr[\chi \in C_u(b)]=\frac{1}{\lcmG^{t}}$. Similarly, it follows that $\mathbb{E}[I_{\chi' \to b+u}]=\frac{1}{\lcmG^{t}}$.
       We thus have
        $\Cov[I_{\chi \to b+u}\cdot I_{\chi' \to b+u}] = \frac{1}{\lcmG^{2t}}-\frac{1}{\lcmG^{2t}} = 0$.
       \end{description}

    \item [Case 2: $\chi,\chi'$ are both nonzero, distinct and one is not a multiple of the other.]
    In this case, \\ $I_{\chi \to b+u}\cdot I_{\chi' \to b+u}=1$ if and only if $\forall i \in [t], \langle \chi, v_i \rangle=b_i+u_i$ and $\langle \chi', v_i \rangle=b_i+u_i$. Fix $i\in [t]$ and $u_i$. Since $\chi'$ is not a multiple of $\chi$, the bucket $\chi'$ belongs to is not determined by that of $\chi$.
   Thus, for each $i$, the  probability that $v_i$ belongs to this bucket is $\frac{1}{\lcmG^2}$. Since $u_1,\ldots,u_t,v_1,\ldots,v_t$ are all independent, we have that $\E[I_{\chi \to b+u}\cdot I_{\chi' \to b+u}]=\frac{1}{\lcmG^{2t}}$.
            
    Now, for the same reason as in the previous case, $\E[I_{\chi \to b+u}]=\E[I_{\chi' \to b+u}]=\frac{1}{\lcmG^t}$. we thus have that $\cov(I_{\chi \to b+u}, I_{\chi' \to b+u}) = \frac{1}{\lcmG^{2t}}-\frac{1}{\lcmG^t}\cdot\frac{1}{\lcmG^t}=0$.
    \end{description}
    
    So, analyzing all possible cases, we observe that $\cov(I_{\chi \to b+u}, I_{\chi' \to b+u}) = 0$. 
    \end{enumerate}
\end{proof}


\color{black}
We will also need the following application of Chebyshev's inequality. 

\begin{lemma} \label{lem:Chebyshev-cov}
    Let $X=\sum_{i \in [n]} X_i$, where the $X_i$’s are real random
variables satisfying -- (i) $0 \le X_i \le \tau$, and (ii) $\mathsf{cov}(X_i,X_j) = 0$, for $i \ne j$. Further, assume that $\mathbb{E}[X] > 0$. Then for any $\epsilon' > 0$, we have $\Var[X] \leq \tau \mathbb{E}[X]$ and 
\[ 
\Pr[X \le (1 - \epsilon')\mathbb{E}[X]]\;\le\;\frac{\tau}{\epsilon'^2 \mathbb{E}[X]}.
\]
\end{lemma}

\begin{proof}
    \begin{align*}
        \Var[X] &= \sum_i \Var[X_i] + \Cov_{i,j:i\neq j} [X_iX_j] \\
        &=\sum_i \mathbb{E}[X_i^2], \text{ since } \Cov_{i,j:i\neq j} [X_iX_j] = 0 \\
        &\leq \sum_i \tau \mathbb{E}[X_i] \\
        &= \tau \mathbb{E}[X].
    \end{align*}
    Therefore, by Chebyshev's inequality, 
    \begin{align*}
        \Pr[X\leq (1-\epsilon')\mathbb{E}[X]] &= \Pr[X-\mathbb{E}[X] \leq -\epsilon'\mathbb{E}[X]] \\
        &\leq \Pr[|X-\mathbb{E}[X]| \geq \epsilon'\mathbb{E}[X]] \\
        &\leq \frac{\Var[X]}{\epsilon'^2 \mathbb{E}[X]^2} \leq \frac{\tau \mathbb{E}[X]}{\epsilon'^2 \mathbb{E}[X]^2} = \frac{\tau}{\epsilon'^2 \mathbb{E}[X]}.
    \end{align*}
\end{proof}
Finally we are ready to prove \cref{lem:bucket}.

\begin{proof}[Proof of \cref{lem:bucket}]
To prove this lemma, let us divide the Fourier coefficients of $f$ in two sets: 
\[
A_1\;:=\;\{\chi: |\widehat{f}(\chi)|^2 \geq \tau\}, \ A_2\;:=\; \{\chi: |\widehat{f}(\chi)|^2 < \tau\}\;.\]

Since $f$ is $\epsilon$-far from any $s$-sparse Boolean function with $\epsilon \leq \frac{1}{8(2 s)^{\varphi(\lcmG)}}$, so by \cref{structure_theorem_2}, it is $\frac{\epsilon}{\sqrt{2}}$-far from any $s$-sparse function. We will prove the lemma in two cases. 

\paragraph*{Case 1: $|A_1|\geq s+1$.}
In this case 
from \cref{thm:random_shift-coset} (part (2)) implies that with probability at  least $\frac{19}{20}$, at least $s+1$ buckets will contain elements from $A_1$. Therefore, the weights of those buckets will be at least $\tau$, and hence it follows. 



\paragraph*{Case 2: $|A_1|\leq s$ \& $\wt(A_2) \geq \epsilon^2/2$.}



Let us define the random variables $Y_b= \wt(C_u(b)\cap A_2)$ for each $b$. Then $Y_b=\sum_{\chi\in A_2} |\widehat{f}(\chi)|^2 \cdot I_{\chi \to b+u}$, where $I_{\chi \to b+u}$ is the indicator random variable which takes the value $1$ if $\chi\in C_u(b)$ and $0$ otherwise. 

To show at least $(s+1)$ of the buckets have weight more than $\tau$ we will first prove that on expectation the weight of a bucket is significantly more than $\tau$. Then we will use Chebyshev's inequality to prove that the probability that a bucket has less than $\tau$ is small and finally, we will use Markov's inequality to upper bound the number of small buckets - or in other words prove that at least $(s+1)$ buckets have weight more than $\tau$.  

We start by proving that on expectation the weights of the buckets are large. Note that,
\begin{equation}\label{eqn_expectation}
    \mathbb{E}[Y_b] = \sum_{\chi\in A_2} |\widehat{f}(\chi)|^2 \Pr[\chi \in C_u(b)] = \frac{1}{\lcmG^t} \sum_{\chi\in A_2} |\widehat{f}(\chi)|^2 \geq \frac{\epsilon^2}{2\lcmG^t},
\end{equation}
where the second equality follows from the fact that $\sum_{\chi\in A_2} |\widehat{f}(\chi)|^2 \geq (\frac{\epsilon}{\sqrt{2}})^2$ and from \cref{thm:random_shift-coset} (part (1)) where we show that $\Pr[\chi \in C_u(b)] = \frac{1}{\lcmG^t}$.

Before we apply Chebyshev's inequality to upper bound the probability that a bucket has small weight we  upper bound the variance of $Y_b$. 
We can write $Y_b= \sum_{\chi\in A_2} Y_{\chi,b}$, where $Y_{\chi,b} = |\widehat{f}(\chi)|^2 \cdot I_{\chi\to b+u}$. Note that $ 0\leq Y_b \leq \tau$, $\mathbb{E}[Y_b]>0$ (by \cref{eqn_expectation}) and $\cov(Y_{\chi,b},Y_{\chi',b})= 0$ (by \cref{thm:random_shift-coset}, where $\chi,\chi'\in A_2$). So, by \cref{lem:Chebyshev-cov}, $$\Var[Y_b] \leq \tau \mathbb{E}[Y_b].$$




    Therefore, by Chebychev's inequality we have that, for any $\epsilon'$, 
    \begin{align*}
        \Pr[Y_b\leq (1-\epsilon')\mathbb{E}[Y_b]] &\leq \frac{\Var[Y_b]}{\epsilon'^2\mathbb{E}[Y_b]^2}  
        \leq \frac{\tau \mathbb{E}[Y_b]}{\epsilon'^2\mathbb{E}[Y_b]^2} = \frac{\tau}{\epsilon'^2\mathbb{E}[Y_b]} \leq \frac{\tau \cdot 2\lcmG^t}{\epsilon'^2\cdot \epsilon^2}.
    \end{align*}

Since the total number of buckets is $\lcmG^t$, therefore the expected number of buckets with weight $\leq (1-\epsilon') \mathbb{E}[Y_b]$ is at most $\lcmG^t \cdot \frac{\tau \cdot 2\lcmG^t}{\epsilon'^2 \epsilon^2}$. Thus by Markov's inequality, with probability at least $\frac{9}{10}$, the number of buckets with weight more that $(1-\epsilon')\mathbb{E}[Y_b]$ is at most $10 \cdot \frac{\tau \cdot 2\lcmG^t}{\epsilon'^2 \epsilon^2}  \cdot \lcmG^t$.

Now, set $\epsilon':=\frac{19}{20}$. Thus, 
$$
    (1-\epsilon')\E[Y_b] \geq \frac{1}{20}\cdot\frac{\epsilon^2}{2\lcmG^t}=\frac{\epsilon^2}{40\lcmG^t}\geq \tau. 
$$
It thus follows that with probability at least $\frac{9}{10}$, the number of buckets with weight more that $(1-\epsilon')\mathbb{E}[Y_b]$ is at most $10 \cdot \frac{\tau \cdot 2\lcmG^t}{\epsilon'^2 \epsilon^2}  \cdot \lcmG^t$.

We now finish by showing that $10 \cdot \frac{\tau \cdot 2\lcmG^t}{\epsilon'^2 \epsilon^2}  \cdot \lcmG^t \leq \lcmG^t-(s+1)$.
\begin{align*}
    &10 \cdot \frac{\tau \cdot 2\lcmG^t}{\epsilon'^2 \epsilon^2} \cdot \lcmG^t \leq \lcmG^t - (s+1) \\
    &\Longleftarrow \tau \leq \frac{\epsilon^2}{\lcmG^t} \cdot \biggl(1-\frac{s+1}{\lcmG^t} \biggr) \cdot \left(\frac{19}{20}\right)^2 \cdot \frac{1}{20} \\
    & \Longleftarrow \tau \leq \frac{\epsilon^2}{\lcmG^t} \cdot \frac{1}{2}  \cdot \frac{1}{20} = \frac{\epsilon^2}{40\lcmG^t}.
\end{align*}
where the last implication follows from the fact that $(1-\frac{s+1}{\lcmG^t}) \geq \frac{1}{2}\cdot \left(\frac{20}{19}\right)^2$ for all $s$, since $t = \lceil \log_\lcmG (20s^2)\rceil +1$. 
We have set $\tau \leq \frac{\epsilon^2}{40\lcmG^t}$ in the algorithm. Hence, the number of heavy buckets is at least $(s+1)$. 

\end{proof}

\section{Degree over \texorpdfstring{$\G$}{}}\label{section_deg_p}

The following definition is a generalization of $\deg_2(f)$, that is degree over $\mathbb{Z}_2$ of a Boolean function $f$, as defined in~\cite{BernasconiC99}.

\begin{definition}\label{defn_deg_p}
    Let $f:\G \to \{-1,+1\}$ be a Boolean valued function. Consider all possible restrictions $f|_{V_{b,r_1, \ldots,r_t}}$ of $f$, where $V_{b,r_1, \ldots,r_t}$ is a coset of $\G$ as defined in \cref{defn_cosets_groups}. Then the degree over $\G$ of $f$, denoted by $\deg_{\G}$, is defined in the following way.
    \begin{align*}
        \deg_{\G}(f) = &\min_t \{\lfloor \log_\lcmG |\G| \rfloor-t : \exists r_1, \ldots, r_t \in \supp(\widehat{f}), b= (b_1,\ldots b_t) \in \mathbb{Z}_\lcmG^t, \\ 
        &\text{ such that } f|_{V_{b,r_1, \ldots,r_t}} \text{ has full sparsity}\},
    \end{align*}
    where $\lcmG = \LCM \{ p_1^{m_1}, \ldots, p_T^{m_T} \}$.
\end{definition}

\begin{remark}\label{remark_dim_supp}
    In \cref{defn_deg_p}, the restriction $f|_{V_{b,r_1, \ldots,r_t}}$ has full sparsity means that $s_{f|_{V_{b,r_1, \ldots,r_t}}} \geq \frac{|\G|}{\lcmG^t} \geq \lcmG^{\lfloor \log_\lcmG |\G| \rfloor-t}$ from \cref{lemma_restriction} and \cref{lemma_cosets_of_V_0}, where $s_{f|_{V_{b,r_1, \ldots,r_t}}}$ is the Fourier sparsity of the function $f|_{V_{b,r_1, \ldots,r_t}}$.
\end{remark}

Next, let us define the $p_i^{m_i} \times p_i^{m_i}$ matrices $V_i$, $i \in [T]$ by
    \begin{align*}
        V_i = 
        \begin{bmatrix}
            1 & 1 & 1 & \cdots & 1 \\
            1 & \omega_{p_i^{m_i}} & \omega_{p_i^{m_i}}^2 & \cdots & \omega_{p_i^{m_i}}^{{p_i^{m_i}}-1} \\
            1 & \omega_{p_i^{m_i}}^2 & \omega_{p_i^{m_i}}^4 & \cdots & \omega_{p_i^{m_i}}^{2({p_i^{m_i}}-1)} \\
            \cdot & \cdot & \cdot & \cdot & \cdot \\
            \cdot & \cdot & \cdot & \cdot & \cdot \\
            \cdot & \cdot & \cdot & \cdot & \cdot \\
            1 & \omega_{p_i^{m_i}}^{{p_i^{m_i}}-1} & \omega_{p_i^{m_i}}^{2({p_i^{m_i}}-1)} & \cdots & \omega_{p_i^{m_i}}^{({p_i^{m_i}}-1)({p_i^{m_i}}-1)} \\
        \end{bmatrix}
    \end{align*}

    Let us consider the Kronecker product (\cref{defn_kronecker}) of these matrices $V_1, \ldots, V_T$, which is a $(p_1^{m_1} \cdots p_T^{m_T}) \times (p_1^{m_1}\cdots p_T^{m_T})$ matrix, and is denoted by $V_1 \otimes \cdots \otimes V_T$.

\begin{lemma}\label{lemma_coeff_V_n}
    Let $f:\G \to \{-1,+1\}$ be a Boolean valued function and $\widehat{f}$ denotes its Fourier transform. Then $\widehat{f} = \frac{1}{p_1^{m_1} \cdots p_T^{m_T}} (V_1 \otimes \cdots \otimes V_T) f$, where $f$ is the $p_1^{m_1} \cdots p_T^{m_T} \times 1$ matrix formed by the values $f(x)$, $x\in \G$ are arranged in usual lexicographic order, and $\widehat{f}$ is the $p_1^{m_1} \cdots p_T^{m_T} \times 1$ matrix formed by the values $\widehat{f}(\chi_r)$, $r \in \G$ are arranged in usual lexicographic order.
\end{lemma}

\begin{proof}
We will use induction on $T$ to prove this.
\begin{description}
    \item[Base case.] When $T=1$, then for each element $i \in \mathbb{Z}_{p_1^{m_1}}$, the Fourier coefficients are given by
    \begin{align*}
        \widehat{f}(\chi_i) = \frac{1}{p_1^{m_1}} \biggl( f(0) + f(1) \omega_{p_1^{m_1}}^i + f(2) \omega_{p_1^{m_1}}^{2i} + \ldots + f(p_1^{m_1}-1) \omega_{p_1^{m_1}}^{(p_1^{m_1}-1)i} \biggr).
    \end{align*}

    Hence $\widehat{f} = \frac{1}{p_1^{m_1}} V_1 f$ in this case.

    \item[Induction step.] Let for all $T \leq K-1$, $\widehat{f} = \frac{1}{p_1^{m_1}\cdots p_{K-1}^{m_{K-1}}} V_1 \otimes \cdots \otimes V_{K-1} f$.

    For $T=K$, for each element $(i,r) \in \mathbb{Z}_{{p_1}^{m_1}} \times \cdots \times \mathbb{Z}_{{p_{K-1}}^{m_{K-1}}} \times \mathbb{Z}_{{p_K}^{m_K}}$, the Fourier coefficients of $f$ are given by
    \begin{align}\label{eqn_coeff_matrix}
        \widehat{f}(\chi_{(i,r)}) = \frac{1}{p_1^{m_1} \cdots p_K^{m_K}} \sum_{x \in \mathbb{Z}_{{p_1}^{m_1}} \times \cdots \times \mathbb{Z}_{{p_{K-1}}^{m_{K-1}}}} &\biggl( \bigl( f(0, x) + f(1, x) \omega_{p_K^{m_K}}^i + f(2,x) \omega_{p_K^{m_K}}^{2i} + \cdots \nonumber \\
        &+ f({p_K^{m_K}}-1,x) \omega_{p_K^{m_K}}^{({p_K^{m_K}}-1)i} \bigr)  \omega_{\LCM \{p_1^{m_1}, \ldots, p_{K-1}^{m_{K-1}}\}}^{r * x} \biggr).
    \end{align}

    From \cref{eqn_coeff_matrix}, any block corresponding to the coefficients $\widehat{f}(\chi_{i,r})$ and the function values $f(j,x)$ is given by the matrix $\omega_{p_K^{m_K}}^{ij} V_1 \otimes \cdots \otimes V_{K-1}$, $r, x \in \mathbb{Z}_{{p_1}^{m_1}} \times \cdots \times \mathbb{Z}_{{p_{K-1}}^{m_{K-1}}}$. Hence, $\widehat{f} = \frac{1}{p_1^{m_1} \cdots p_K^{m_K}} V_1 \otimes \cdots \otimes V_K f$.
\end{description}

Thus by induction it follows that $\widehat{f} = \frac{1}{p_1^{m_1} \cdots p_T^{m_T}} V_1 \otimes \cdots \otimes V_T f$ for any positive integer $T$.
\end{proof}

One can equivalently define $\deg_{\G}$ in the following way.

\begin{definition}\label{defn_2_deg_p}
    Let $f:\G \to \{-1,+1\}$ be a Boolean valued function. Let us arrange the elements $x \in \G$ in lexicographic order along the rows and columns of the matrix $V_1\otimes \cdots \otimes V_T$. Consider all possible restrictions $f|_{V_{b,r_1, \ldots,r_t}}$ of $f$, where $V_{b,r_1, \ldots,r_t}$ is a coset of $\G$ as defined in \cref{defn_cosets_groups}. Then the degree over $\G$ of $f$ is defined as the dimension of the largest coset $V_{b,r_1, \ldots,r_t}$ such that the restriction $f|_{V_{b,r_1, \ldots,r_t}}$ takes the value $1$ at those points such that each entry of the row/column matrix formed by the sum of the rows/columns corresponding to the elements where $f|_{V_{b,r_1, \ldots,r_t}}$ takes the value $1$ is $\neq 0$. 
    
    [Note that in that case all the Fourier coefficients are nonzero.]
\end{definition}

\begin{remark}
    Observe, from the Vandermonde matrix $V_1\otimes \cdots \otimes V_T$, that, any character, (except $\chi_0$, the empty character), of $f$ can be written as $\frac{1}{p_1^{m_1} \cdots p_T^{m_T}} \biggl( \sum_i c_i\omega_\lcmG^i -\sum_j d_j \omega_\lcmG^j \biggr)$, where $c_i, d_j$ are non-negative integers for all $i,j$ and $\omega_\lcmG$ is a primitive root of unity, $\lcmG = \LCM \{p_1^{m_1}, \ldots, p_T^{m_T}\}$. (Any root of unity of the form $\omega_{p_1^{m_1}} \cdots \omega_{p_T^{m_T}}$ is also a $\omega_\lcmG^{th}$ root of unity.) Note that the sum of elements of each row/column of $V_1\otimes \cdots \otimes V_T$ is equal to $0$. That is, $\sum_i c_i\omega_\lcmG^i + \sum_j d_j \omega_\lcmG^j =0$.
    So, for a column with $\sum_i c_i\omega_\lcmG^i \neq 0$, we have the corresponding character $= \frac{1}{p_1^{m_1}\cdots p_T^{m_T}} \biggl( \sum_i c_i\omega_n^i -\sum_j d_j \omega_\lcmG^j \biggr) = \frac{1}{p_1^{m_1}\cdots p_T^{m_T}} \biggl( \sum_i c_i\omega_\lcmG^i + \sum_i c_i\omega_\lcmG^i \biggr) \neq 0$. Hence \cref{defn_2_deg_p} follows.

    Note that, the first row/column, corresponding to the constant coefficient, is not equal to $0$. In the case when $\G$ is of prime order,
    the constant coefficient is always $\neq 0$, as $f$ can take the value $1$ at an odd number of elements and the value $-1$ at an even number of elements, or vice versa. In case of general Abelian groups, this may not be true, so the restriction has to be chosen in such a way that it does not take the value $1$ at half of the points, so that the constant coefficient is not equal to $0$. 
\end{remark}

\section{Query lower bound on adaptively testing sparsity}\label{section_algo_lower_bound}

\begin{theorem}[Theorem~\ref{thm_testing_lower_bound}, restated]
    For Boolean valued functions on $\G$, to adaptively test $s$-sparsity, the query lower bound of any algorithm is $\Omega(\sqrt{s})$.
\end{theorem}

In this section, we will prove \cref{thm_testing_lower_bound}, that is, find a lower bound on the number of queries for our testing algorithm.


Let $\tau>0$ and $C$ be a constant dependent on $\tau$, and will be determined later. 
Let $H$ be a random subgroup (see \cref{defn_random_subgroup}) of $\mathbb{Z}_{p_{i_1}^{m_{i_1}}} \times \cdots \times \mathbb{Z}_{p_{i_{Ct}}^{m_{i_{Ct}}}}$ of co-dimension $t$ (see \cref{defn_codimension}), where $i_1, \ldots, i_{Ct} \in \{1, \ldots, T\}$, and $\mathcal{C}$ is the set of all cosets of $H$. That is, $(H, \mathcal{C})$ forms a random $t$-dimensional coset structure (see \cref{defn_t_dim_coset}). Now, let us define the following two probability distributions:
\begin{itemize}
    \item Let us construct random functions $f$ by making $f$ constant on each coset of $\mathcal{C}$, the constant being chosen randomly from $\{-1,+1\}$. This implies that the dimension of $f$ is $t$. Let us call this probability distribution \textbf{$\mathcal{D}_{yes}$}.

    \item We choose a random function $f$ randomly from $\mathbb{Z}_{p_{i_1}^{m_{i_1}}} \times \cdots \times \mathbb{Z}_{p_{i_{Ct}}^{m_{i_{Ct}}}}$, conditioned on the fact that $f$ is $2 - \tau$ far in $\ell_2$ from any function which has $\deg_{\G} = t$. Let us call this distribution \textbf{$\mathcal{D}_{No}$}.
\end{itemize}

We will prove the following theorem first, \cref{thm_testing_lower_bound} follows as a corollary of this.

\begin{theorem}\label{thm_algo_lower_bound}
    Any adaptive query algorithm that distinguishes $\mathcal{D}_{Yes}$ and $\mathcal{D}_{No}$ with probability $\geq \frac{1}{3}$ has to make at least $\Omega(\lcmG^{\frac{t}{2}})$ queries. 
\end{theorem}

The following lemma will be useful in proving \cref{thm_algo_lower_bound}.

\begin{lemma}\label{lemma_dim_sparsity_deg}
    For any Boolean valued function $f:\G \to \{-1,+1\}$, $$\lcmG^{r(f)} \geq s_f \geq \lcmG^{\deg_{\G}(f)}.$$
\end{lemma}

\begin{proof}
    The proof of the first inequality follows trivially from the definitions of sparsity $s_f$ (\cref{def:fourier-sparsity}) and dimension $r(f)$ (\cref{defn_Fourier_dimension}) of $f$.
    
    The proof of the second inequality follows from the definition of $\deg_{\G}(f)$ (\cref{defn_deg_p}) and \cref{remark_dim_supp}, since $s_f \geq s_{f|_{V_{b,r_1, \ldots,r_t}}} \geq \lcmG^{\lfloor \log_\lcmG |\G| \rfloor - t} = \lcmG^{\deg_\G(f)}$, where $V_{b,r_1, \ldots, r_t}$ is the coset, with $t$ being the least possible nonnegative integer, such that $f|_{V_{b,r_1, \ldots,r_t}}$ has full sparsity.
\end{proof}

\begin{corollary}[Corollary of \cref{lemma_dim_sparsity_deg}]
    The set of functions with $r(f) = t$ is a subset of the set of functions with sparsity $s=\lcmG^t$, which in turn is a subset of the set of functions with $\deg_{\G} = t$.
\end{corollary}

\cref{thm_testing_lower_bound} directly follows from \cref{thm_algo_lower_bound} and \cref{lemma_dim_sparsity_deg}.

We will show that if an adaptive query algorithm makes less than $q< \Omega(\lcmG^{t/2})$ queries, then the total variation distance $||\mathcal{D}_{Yes} - \mathcal{D}_{No}||_{TV}$ between the two distributions $\mathcal{D}_{Yes}$ and $\mathcal{D}_{No}$ is $\leq \frac{1}{3}$, where the total variation distance is as defined in \cref{defn_total_variation} and \cref{lemma_total_variation}. This proves that any adaptive query algorithm which distinguishes between $\mathcal{D}_{Yes}$ and $\mathcal{D}_{No}$, that is, for which $||\mathcal{D}_{Yes} - \mathcal{D}_{No}||_{TV} > \frac{1}{3}$ must make at least $\Omega(\lcmG^{t/2})$ queries.

Any adaptive query algorithm that makes $q$ queries can be thought of as a $n'$-ary decision tree, whose internal nodes are elements from $\mathbb{Z}_{p_{i_1}^{m_{i_1}}} \times \cdots \times \mathbb{Z}_{p_{i_{Ct}}^{m_{i_{Ct}}}}$ and the leaves are labeled by \textbf{accept} or \textbf{reject}, where $n' = \LCM\{ p_{i_1}^{m_{i_1}}, \ldots, p_{i_{Ct}}^{m_{i_{Ct}}} \}$. Note that $n'$ can be at most $\lcmG$.
Also, the length of each path from the root to the leaves of the tree is equal to $q$ and no element repeats in the tree, that is, it does not appear more than once in the tree. 

Let $\mathcal{L}_{Yes}$ and $\mathcal{L}_{No}$ be the distributions on the leaves of the decision tree from $\mathcal{D}_{Yes}$ and $\mathcal{D}_{No}$ respectively. We will show that $||\mathcal{L}_{Yes} - \mathcal{L}_{No}||_{TV} \leq \frac{1}{3}$ when the number of queries $q \leq \frac{1}{\sqrt{20}} \lcmG^{\frac{t}{2}}$.

\begin{lemma}\label{lemma_uniform_vs_yes}
    $||\mathcal{L}_{Yes}- \mathcal{L}_{Uniform}|| \leq \frac{1}{19}$.
\end{lemma}

\begin{proof}[Proof of \cref{lemma_uniform_vs_yes}]
    Consider the set of all functions of dimension $t$. Let the depth $q$ of the decision tree corresponding to our algorithm is $\leq \frac{1}{\sqrt{20}} \lcmG^{\frac{t}{2}}$. Then $t \geq 2\log_\lcmG q + \log_\lcmG 20$. So by \cref{thm:random-coset} (part 2), all the characters on a particular path fall into different buckets with probability $\geq (1- \frac{1}{20})$. Let us choose a function $f$ randomly from $\mathcal{D}_{Yes}$. For each leaf $\ell$, let us consider the characters on the path from root to $\ell$. Since the length of the path is $q$, so with probability $1-\frac{1}{20}$, the characters on the path belongs to different buckets. Then the probability that $f$ is consistent with the path to $\ell$, given that all the characters on the path from root to $\ell$ belong to different buckets is $= \frac{1}{n'^q}$. Note that $n' \leq n$. Therefore, the probability that $f$ is consistent with the path to $\ell$ is $\geq (1- \frac{1}{20})\frac{1}{n'^q} = \frac{19}{20} \frac{1}{n'^q}$. Let us denote this probability by $\mathcal{P}_{Yes}$.

    Now, for each path from root to $\ell$, we know that $\mathcal{P}_{Uniform}(x)= \frac{1}{n'^q}$, where $x$ is a character in the path from root to $\ell$. Then, 
    \begin{align*}
        &|| \mathcal{P}_{Yes} - \frac{19}{20} \mathcal{P}_{Uniform} ||_1 \\
        &= \sum_{\{x \in \text{ the path to } \ell\}} |\mathcal{P}_{Yes}(x) - \frac{19}{20} \mathcal{P}_{Uniform} (x)| \\
        &= \sum_{\{x \in \text{ the path to } \ell\}} |\mathcal{P}_{Yes}(x) - \frac{19}{20} \frac{1}{n'^q}| \\
        &= \sum_{\{x \in \text{ the path to } \ell : \mathcal{P}_{Yes}(x) \geq \frac{19}{20} \frac{1}{n'^q}\}} |\mathcal{P}_{Yes}(x) - \frac{19}{20} \frac{1}{n'^q}| &\text{since } \mathcal{P}_{Yes}(x) \geq \frac{19}{20} \frac{1}{n'^q} \\
        &\leq \biggl[ \sum_{\{x \in \text{ the path to } \ell : \mathcal{P}_{Yes}(x) \geq \frac{19}{20} \frac{1}{n'^q}\}} \mathcal{P}_{Yes}(x) \biggr] - \frac{19}{20} \\
        &\leq 1- \frac{19}{20} = \frac{1}{20}.
    \end{align*}
\end{proof}

Therefore,
\begin{align*}
    ||\mathcal{P}_{Yes}-\mathcal{P}_{Uniform}||_1 
    &\leq ||\mathcal{P}_{Yes}- \frac{\mathcal{P}_{Yes}}{\frac{19}{20}}||_1 + ||\frac{\mathcal{P}_{Yes}}{\frac{19}{20}} -\mathcal{P}_{Uniform}||_1 \\
    &\leq \frac{1}{19} ||\mathcal{P}_{Yes}||_1 + \frac{\frac{1}{20}}{\frac{19}{20}} \\
    &= \frac{2}{19} &\text{since } ||\mathcal{P}_{Yes}||_1 =1.
\end{align*}

Hence, 
\begin{align*}
    ||\mathcal{P}_{Yes}-\mathcal{P}_{Uniform}||_{TV} = \frac{1}{2} ||\mathcal{P}_{Yes}-\mathcal{P}_{Uniform}||_1 = \frac{1}{19}.
\end{align*}

\begin{lemma}\label{lemma_uniform_vs_no}
    Let $\mathcal{D}_{Uniform}$ be the uniform distribution on the set of functions from $\mathbb{Z}_{p_{i_1}^{m_{i_1}}} \times \cdots \times \mathbb{Z}_{p_{i_{Ct}}^{m_{i_{Ct}}}}$ to $\{-1,+1\}$. Then, the probability that the degree over $\G$ of a randomly chosen function $f$ is $(2 - \tau)$ close to $t$ in $\ell_2$ is $\leq \frac{1}{10}$.
\end{lemma}

\begin{proof}[Proof of \cref{lemma_uniform_vs_no}]
    Let us fix a function $g$ from $\mathbb{Z}_{p_{i_1}^{m_{i_1}}} \times \cdots \times \mathbb{Z}_{p_{i_{Ct}}^{m_{i_{Ct}}}}$ to $\{-1,+1\}$ such that $\deg_{\G}(g) = t$. Let $f$ be  randomly chosen function, where $f$ is also from $\mathbb{Z}_{p_{i_1}^{m_{i_1}}} \times \cdots \times \mathbb{Z}_{p_{i_{Ct}}^{m_{i_{Ct}}}}$ to $\{-1,+1\}$. We need to show that the probability that the $\ell_2$ distance between $f$ and $g$ is $\leq$ $(2 - \tau)$ is $\leq \frac{1}{10}$.

    Let $\mathbb{I}_x$ be the indicator variable that takes the value $1$ when $f(x) \neq g(x)$, and it is $0$ otherwise. Then $X_g' = \frac{\sum_{x \in \mathbb{Z}_{p_{i_1}^{m_{i_1}}} \times \cdots \times \mathbb{Z}_{p_{i_{Ct}}^{m_{i_{Ct}}}}} \mathbb{I}_x}{p_{i_1}^{m_{i_1}} \cdots p_{i_{Ct}}^{m_{i_{Ct}}}}$ denotes the Hamming distance between them. Also, by \cref{lemma_Boolean_distance_equivalence}, if $X_g$ denotes thesquare of the $\ell_2$ distance between $f$ and $g$, then $X_g = 4 X_g'$. Now,
    \begin{align*}
        &\mathbb{E}[X_g'] = \frac{\sum_{x \in \mathbb{Z}_{p_{i_1}^{m_{i_1}}} \times \cdots \times \mathbb{Z}_{p_{i_{Ct}}^{m_{i_{Ct}}}}} \mathbb{E}[\mathbb{I}_x]}{p_{i_1}^{m_{i_1}} \cdots p_{i_{Ct}}^{m_{i_{Ct}}}} = \frac{1}{p_{i_1}^{m_{i_1}} \cdots p_{i_{Ct}}^{m_{i_{Ct}}}} \frac{1}{2} p_{i_1}^{m_{i_1}} \cdots p_{i_{Ct}}^{m_{i_{Ct}}} = \frac{1}{2} \\
        &\Rightarrow \mathbb{E}[X_g] = 4 \mathbb{E}[X_g'] = 2.
    \end{align*}

    Therefore, by Chernoff (\cref{Chernoff}),
    \begin{align*}
        \Pr[X_g' \leq \frac{1}{2} - \frac{\tau}{4}] 
        &\leq \Pr[|X_g' -\mathbb{E}[X_g']| \geq \frac{\tau}{4}] \\
        &= \Pr[|\sum_{x \in \mathbb{Z}_{p_{i_1}^{m_{i_1}}} \times \cdots \times \mathbb{Z}_{p_{i_{Ct}}^{m_{i_{Ct}}}}} \mathbb{I}_x - \mathbb{E}[\sum_{x \in \mathbb{Z}_{p_{i_1}^{m_{i_1}}} \times \cdots \times \mathbb{Z}_{p_{i_{Ct}}^{m_{i_{Ct}}}}} \mathbb{I}_x]| \geq \frac{\tau}{4} \times p_{i_1}^{m_{i_1}} \cdots p_{i_{Ct}}^{m_{i_{Ct}}}] \\
        &\leq 2 \exp (-\frac{2}{3} \bigl(\frac{\tau}{4} \bigr)^2 p_{i_1}^{m_{i_1}} \cdots p_{i_{Ct}}^{m_{i_{Ct}}}),
    \end{align*}

    which implies that 
    \begin{align*}
        \Pr[X_g \leq 2 - \tau] \leq 2 \exp (-\frac{1}{24} \tau^2 p_{i_1}^{m_{i_1}} \cdots p_{i_{Ct}}^{m_{i_{Ct}}}).
    \end{align*}


Now, the number of functions with degree over $\G$ equal to $t$ is $\leq n'^{\binom{(C+1)t}{t}}$, where $n' = \LCM \{ p_{i_1}^{m_{i_1}}, \ldots, p_{i_{Ct}}^{m_{i_{Ct}}} \}$.

The probability that $f$ is $(2 - \tau)$-close to a degree $t$ function in $\ell_2$ is given by,
\begin{align*}
    &\Pr[\cup_{g:\deg_{\G}(g)=t} \{X_g \leq 2 - \tau\}] \\
    &\leq n'^{\binom{(C+1)t}{t}} \times 2 \exp (-\frac{2}{3} \tau^2 p_{i_1}^{m_{i_1}} \cdots p_{i_{Ct}}^{m_{i_{Ct}}}) \\
    &= 2 \exp \biggl( \binom{(C+1)t}{t} \ln n' - \ln t! - \frac{2}{3} \tau^2 p_{i_1}^{m_{i_1}} \cdots p_{i_{Ct}}^{m_{i_{Ct}}} \biggr) \\
    &\leq 2 \exp \biggl( \binom{(C+1)t}{t} \ln n' - \ln t! - \frac{2}{3} \tau^2 (\min \{p_{i_1^{m_{i_1}}}, \ldots, p_{i_{Ct}^{n_{i_{Ct}}}}\})^{Ct} \biggr) \\
    &\leq \frac{1}{10}, 
\end{align*}
taking $C \geq \frac{1}{t\ln \min \{p_{i_1^{m_{i_1}}}, \ldots, p_{i_{Ct}^{m_{i_{Ct}}}}\}} \ln \biggl( \frac{3(\binom{(C+1)t}{t} \ln n' - \ln t! +\ln 20)}{2\tau^2} \biggr)-1$.
\end{proof}

\begin{corollary}\label{cor_no_uniform_leq_prob}
    $||\mathcal{L}_{No}- \mathcal{L}_{Uniform}||_{TV} \leq \frac{1}{10}$.
\end{corollary}

\begin{proof}
Let $A$ be the set of all functions $f:\mathbb{Z}_{p_{i_1}^{n_{i_1}}} \times \cdots \times \mathbb{Z}_{p_{i_{Ct}}^{n_{i_{Ct}}}} \to \{-1,+1\}$ which are $(2- \tau)$ close to $\deg_{\G} = t$ in $\ell_2$. Then $\mathcal{P}_{No} (A) =0$, and by \cref{lemma_uniform_vs_no}, $\mathcal{P}_{Uniform(A)} \leq \frac{1}{10}$, where $\mathcal{P}_{No}$ and $\mathcal{P}_{Uniform}$ denote the probability over the distributions $\mathcal{D}_{No}$ and $\mathcal{D}_{Uniform}$ respectively. So, $|\mathcal{L}_{No} (A)- \mathcal{L}_{Uniform}(A)| \leq \frac{1}{10}$.

Hence, $$||\mathcal{L}_{No}- \mathcal{L}_{Uniform}||_{TV} \leq \frac{1}{10}.$$
\end{proof}

\begin{proof}[Proof of \cref{thm_algo_lower_bound}]
    From \cref{lemma_uniform_vs_yes}, we have $||\mathcal{L}_{Yes}- \mathcal{L}_{Uniform}|| \leq \frac{1}{19}$. Also, from \cref{cor_no_uniform_leq_prob}, we have $||\mathcal{L}_{No}- \mathcal{L}_{Uniform}||_{TV} \leq \frac{1}{10}$. Therefore, by triangle inequality, 
    \begin{align*}
        &||\mathcal{L}_{Yes} - \mathcal{L}_{No}||_{TV} \\
        &\leq ||\mathcal{L}_{Yes}- \mathcal{L}_{Uniform}||_{TV} + ||\mathcal{L}_{No}- \mathcal{L}_{Uniform}||_{TV} \\
        &\leq \frac{1}{19} + \frac{1}{10} \leq \frac{1}{3}.
    \end{align*}
    
    Hence,
    \begin{align*}
        ||\mathcal{D}_{Yes} - \mathcal{D}_{No}||_{TV} = ||\mathcal{L}_{Yes} - \mathcal{L}_{No}||_{TV} \leq \frac{1}{3}.
    \end{align*}
\end{proof}

\section{Conclusion}
\label{section_Conclusion}


Gopalan et~al.~\cite{gopalan2011testing} was the first to study the problem of testing Fourier sparsity of Boolean function over $\Z^{n}_{2}$. Along the way, they were able to drive fundamental properties of Boolean functions over $\Z_{2}^{n}$, like {\em Granularity} of the Fourier spectrum, that have found many other applications~\cite{ArunachalamCLPW21, fsttcs/0001MMMPS21}. 
In this work, we have extended their results for finite Abelian groups.

Finally, we ask whether it is possible to show a better ($\lcmG$-{\em dependent}) lower bound over $\G$ on the query-complexity of any adaptive sparsity testing algorithm.

\bibliographystyle{alpha}
\bibliography{bibliography}

\end{document}